\documentclass[twoside,11pt]{article}

\usepackage[latin1]{inputenc}
\usepackage{amsmath}
\usepackage{amssymb}
\usepackage{graphicx}
\usepackage{float}
\usepackage[top=3cm, bottom=3cm, left=3cm, right=3cm]{geometry}
\usepackage{hyperref}
\usepackage{subfig}
\usepackage{multirow}
\usepackage{color}
\usepackage{version}
\usepackage{mathabx}

\usepackage{amsthm}
\usepackage{mathrsfs}

\def\tr{{\rm Tr}}
\def\e{{\rm e}}
\def\sp{{\rm sp}}
\def\Ent{{\rm Ent}}

\newcommand{\mb}{\mathbb}
\newcommand{\mc}{\mathcal}
\newcommand{\Hi}{\mathcal{H}}
\newcommand{\En}{\mathcal{E}}
\newcommand{\Sy}{\mathcal{S}}
\newcommand{\Lj}{\mathcal{L}}
\newcommand{\Res}{\mathcal{R}}

\def\un{{\mathchoice {\rm 1\mskip-4mu l} {\rm 1\mskip-4mu l} {\rm 1\mskip-4.5mu l} {\rm 1\mskip-5mu l}}}

\newcommand{\dr}{\partial_}
\newcommand{\al}{\alpha}
\newcommand{\bs}{\boldsymbol}

\newcommand{\mean}{\overline}
\newcommand{\Bo}{\mathcal{B}}
\newcommand{\s}{\sharp}
\newcommand{\R}{\mathbb{R}}
\newcommand{\N}{\mathbb{N}}
\newcommand{\U}{\mathcal{U}}
\renewcommand{\P}{\mathbb{P}}
\newcommand{\E}{\mathbb{E}}
\newcommand{\C}{\mathbb{C}}
\newcommand{\bj}{{\boldsymbol{j}}}
\newcommand{\bfs}{{\boldsymbol{s}}}
\newcommand{\Q}{\mathcal{Q}}
\newcommand{\bfzeta}{{\boldsymbol{\zeta}}}
\newcommand{\refer}{{\rm ref}}
\newcommand{\ds}{\displaystyle}
\newcommand{\rev}{{\rm rev}}

 \def\cB{{\mathcal B}} 
\def\cD{{\mathcal D}} \def\cE{{\mathcal E}} 
 \def\cH{{\mathcal H}} 
  \def\cL{{\mathcal L}}
\def\cM{{\mathcal M}}  
\def\cP{{\mathcal P}}  \def\cR{{\mathcal R}}
\def\cS{{\mathcal S}}

\theoremstyle{definition}

\theoremstyle{plain}
\newtheorem{defi}{Definition}[section]

\newtheorem{theo}{Theorem}[section]
\newtheorem*{hypo}{Assumption}

\newtheorem{prop}[theo]{Proposition}
\newtheorem{lemm}[theo]{Lemma}
\newtheorem{coro}[theo]{Corollary}

\theoremstyle{remark}
\newtheorem{rema}{Remark}[section]

\numberwithin{equation}{section}

\title{ 
		\textbf{Linear response theory and entropic fluctuations in repeated interaction quantum systems}
}

\author{Jean-Fran\c{c}ois Bougron$^{1,2}$, Laurent Bruneau$^{1}$
\\ \\ 
$^1$ D\'epartement de Math\'ematiques and UMR 8088,
CNRS and CY Cergy Paris Universit\'e\\
95000 Cergy-Pontoise, France. Email: laurent.bruneau@u-cergy.fr
\\ \\
$^2$ Univ. Grenoble Alpes, 
CNRS, Institut Fourier\\ 
38000 Grenoble, France. Email: jean-francois.bougron@univ-grenoble-alpes.fr
\\ 
}

\begin{document}

\maketitle

\begin{abstract}
We study Linear Response Theory and Entropic Fluctuations of finite dimensional non-equilibrium repeated interaction systems (RIS). More precisely, in a situation where the temperatures of the probes can take a finite number of different values, we prove analogs of the Green-Kubo fluctuation-dissipation formula and 
Onsager reciprocity relations on energy flux observables. Then we prove a Large Deviation Principle, or Fluctuation Theorem, and a Central Limit Theorem on the full counting statistics of 
entropy fluxes. We consider two types of non-equilibrium RIS: either the temperatures of the probes are deterministic and arrive in a cyclic way, or the temperatures of the probes are described by a sequence of i.i.d. random variables with uniform distribution over a finite set.
\end{abstract}


\section{Introduction}

In this paper we are interested in the linear response theory and entropic fluctuation for a particular class of open quantum systems called Repeated Interaction Systems (RIS), see Section \ref{ris} for a precise description. Our study fits in the wider framework of non-equilibrium quantum statistical mechanics. In this context, linear response theory and entropic fluctuation have attracted lot of attention in the last decades, see e.g. \cite{LS,DS,JOP1,JOP2,JOP3,JOP4,JPP,dRM,dR,JOPP,BPR} and references therein.

Repeated interaction systems consist in a small system $\cS$ coupled to an environment made of a chain of independent probes $\cE^n$ with which $\cS$ will interact in a sequential way, i.e. $\cS$ interacts with $\cE^1$ during the time interval $[0,\tau_1[$, then with $\cE^2$ during the interval $[\tau_1,\tau_1+\tau_2[$, etc. While $\cS$ interacts with a given probe $\cE^n$ the other ones evolve freely according to their intrinsic (uncoupled) dynamics. Formally, if $H_\cS$, $H_{\cE^n}$, denote the non-interacting hamiltonians of $\cS$ and the $\cE^n$'s and $V_n$ denotes the coupling operator between $\cS$ and $\cE^n$ then the hamiltonian of the full system is the time-dependent, piecewise constant, operator
\begin{equation}\label{intro:rishamiltonian}
H(t):= H_\cS+H_{\cE^n}+V_n + \sum_{p\neq n} H_{\cE^p}, \quad t\in[\tau_1+\cdots+\tau_{n-1}, \tau_1+\cdots+\tau_{n}[.
\end{equation}

In the simplest case all the probes are identical, that is each $\cE^n$ is a copy of the same $\cE$ with the same initial state $\rho_\cE$, e.g. a thermal state, and interacts with $\cS$ by means of the same coupling operator $V$ on $\cS + \cE$ for the same duration $\tau$. The dynamics restricted to the small system is shown to be determined by a map $\cL$, see (\ref{def:rdm}), which assigns $\rho(\tau)=\cL(\rho)$ to $\rho$ as the result of the interaction of $\cS$ with one subsystem $\cE$ for the duration $\tau$. Heuristically, from the point of view of the small system, all subsystems interacting in sequence with $\cS$ are equivalent, so that the result of $n\in\N$ repeated interactions amounts to iterating  $n$ times the map $\cL$ on the initial condition $\rho_\cS$. This expresses the Markovian character of repeated interactions in discrete time, see e.g. \cite{BJM} for an introduction to these RIS. 

The typical physical situation of repeated interaction models is that of the one atom maser, see e.g. \cite{FJM, MWM, RBH, WBKM}. Here, $\cS$ is the quantized electromagnetic field of a cavity through which a beam of atoms, the $\cE^n$'s, is shot in such a way that no more than one atom is present in the cavity at any time. Such systems play a fundamental role in the experimental and theoretical investigations of basic matter-radiation processes. On the mathematical side various aspects of RIS have been studied in the literature, see e.g. \cite{KM,AP,AJ,BJMas,BJMrd,HJPR}. We mention also \cite{BP,Bru14} for the analysis of a specific model related to the one-atom maser. 

In order to consider a non-equilibrium situation we shall naturally consider here the more interesting situation where the probes are not always identical but allowed to vary over a finite set, especially via the temperatures of their initial states $\rho_{\cE^n}$ . The picture one should have in mind is that the system $\cS$ is coupled to finitely many reservoirs $\cR_1,\ldots,\cR_M$ ($M\geq 2$) which are initially in thermal equilibrium at possibly different temperatures. The various probes are then associated to one of these $M$ reservoirs, all the probes of a given reservoir being identical. After $n$ interactions the state of $\cS$ is thus given by 
\begin{equation}\label{eq:introrisevolution}
\rho_n=\cL_{j_n}\circ\cdots\circ\cL_{j_1}(\rho)
\end{equation} 
where $j_1,\ldots,j_n\in\{1,\ldots,M\}$ describes the \emph{ordered} sequence of the probes. 

We are interested in the Linear Response Theory for the heat fluxes out (Green-Kubo formula, Onsager reciprocity relations and Central Limit theorem) as well as in the fluctuation of entropy production and the so-called fluctuation relations. The latter can be understood as a generalization of the Green-Kubo formula and Onsager relations for systems far from equilibrium \cite{Ga}. In the framework of quantum dynamical semigroups $\cL_t:= \e^{tL}$, $L=L_1+\cdots+L_M$ and where $L_j$ is the Lindblad generator describing the interaction with the $j$-th reservoir, those questions have been initiated in \cite{LS} and then more recently studied in \cite{DdRM,dRM,JPW}. Eq. (\ref{eq:introrisevolution}) suggests that the situation should be very similar, if not identical, for RIS except that we have a discrete-time dynamics. As we shall see this is only partly correct. If the global strategy of the proofs largely follows those in \cite{LS,JPW}, RIS however have several specificities that have to be taken into account.

The first peculiarity can easily be seen from (\ref{eq:introrisevolution}). As mentioned above there is a precise order in which $\cS$ interacts with the various probes. In order to make sure that $\cS$ interacts as much with each reservoir the first idea is to make the order of interactions cyclic: $\cS$ interacts first with a probe associated to $\cR_1$, then to $\cR_2$ and so on up to $\cR_M$ and then $\cR_1$ again etc. In this case, if it is possible to derive a Green-Kubo type formula we will see that the usual Onsager reciprocity relation will fail, and similarly for the fluctuation relation of entropy production. The reason is simply that the cyclic order of interactions breaks time-reversal invariance even if we suppose that each interaction is time-reversal invariant. Due to the cyclic order of the interactions it is for example reasonable to expect that a change of temperature in $\cR_1$ will have a greater influence on the heat flux in $\cR_2$ than a change of temperature in $\cR_2$ will have on the flux in $\cR_1$. To remedy this lack of global time-reversal invariance we shall therefore also consider the situation where the probes associated to the various reservoirs interact with $\cS$ in a random order. To make it simple we shall consider here the case where the probes are chosen independently at each time and with a uniform distribution (so that on average $\cS$ interacts as much with each reservoir). Our results can be easily generalized to the case of an arbitrary i.i.d. distribution. The case where the distribution of the probes is given by a Markov process will be considered in \cite{BJP}.

The second specificity is related to the fact that the RIS hamiltonian (\ref{intro:rishamiltonian}) is time-dependent. As a consequence, even in the ideal case where all the probes are identical, this may lead to a non-vanishing of entropy production which is usually considered as the signature of a non-equilibrium situation. In this paper we are interested in the response of the system to the presence of thermal forces. This forces us to impose some extra assumption called Assumption (\nameref{Non-Entanglement}), see Section \ref{ssec:NE}, which guarantees that the case where all the temperatures of the various reservoirs are equal is indeed an equilibrium situation. A general linear response theory for RIS should also take into account a departure from this assumption. It is however not clear how to quantify this or, said differently, what is the natural quantity one can associate to a generic RIS and the vanishing of which would correspond to the fulfilment of Assumption (\nameref{Non-Entanglement}), and we postpone this question to future work. Regarding this point we finally mention that the question of linear response theory for time-dependent quantum hamiltonians have been considered in \cite{DS}, in a weak-coupling regime, but there only the perturbation was time-dependent contrary to what happens in RIS.

Finally, to avoid technical issues we shall stick here to the case where all the subsystem's Hilbert spaces, for $\cS$ and the $\cE^n$'s, are finite dimensional. Most results can easily be extended to infinite dimension provided the various assumptions are adapted in an \emph{ad hoc} way, in particular the ergodic Assumption (\nameref{ergodicity}) of Section \ref{sec:ergodicity} has to require the existence of a spectral gap. It is however difficult to find physically relevant models to which these assumptions apply. The model for the one-atom maser studied in \cite{BP,Bru14} does not have a spectral gap for example. It is however possible to still prove the Green-Kubo formula and Onsager relations for this model, see \cite{Bo}. We also mention \cite{BDBP} which considers a RIS type model for the motion of a tight-binding electron and where an analog of the fluctuation relation is proven for the position increments of the electron.

The paper is organized as follows. In Section \ref{Preliminaries} we briefly recall some basic concepts of open systems. In Section \ref{ris} we describe the non-equilibrium RIS model and state the various assumptions which we will use. We will in particular describe in more detail the above mentioned Assumption (\nameref{Non-Entanglement}) and discuss its origin and some of its consequences. Section \ref{lrtentropy} is devoted to the linear response theory and in particular the derivation of a Green-Kubo formula and of Onsager relations. These are stated in Theorem \ref{gko}.  In Section \ref{fluc} we consider the fluctuation of entropy production and prove a fluctuation relation for RIS. We also complement linear response theory with a Central Limit Theorem. Our main results in this section are Theorems \ref{theojpw}, \ref{momentstheo}, \ref{flucresults} and \ref{clt}. Finally the proofs are given in Section \ref{proofs}.

\medskip

\noindent {\bf Acknowledgements.} This research was supported by the \emph{Agence Nationale de la Recherche} through the grant NONSTOPS (ANR-17-CE40-0006) and by the \emph{Initiative d'excellence Paris-Seine}. The research of JFB is partially funded by the Cross Disciplinary Program ``Quantum Engineering Grenoble''. LB warmly thanks UMI-CRM of Montreal for financial support and McGill University for its hospitality during an earlier stage of this work.


\section{Preliminaries}\label{Preliminaries}


\subsection{Observables, states and their evolution}

Throughout the paper $\Hi$ will denote a finite dimensional Hilbert space. We denote by $\mc{B} (\Hi)$ and $\mc{B}_1(\Hi)$ the spaces of bounded and trace-class linear operators on $\Hi$, respectively. We denote by $\un \in \mc{B} (\Hi)$ the identity operator and $\mc{B}^+ (\Hi)$ and $\mc{B}_1^+ (\Hi)$ the sets of positive elements of $\mc{B} (\Hi)$ and $\mc{B}_1 (\Hi)$ respectively. An observable on $\Hi$ is a self-adjoint element in $\mc{B} (\Hi)$. States on $\Hi$ are positive, unital and linear functionals on $\mc{B}(\Hi)$ which are identified with positive trace one elements of $\mc{B}_1 (\Hi)$: for $\rho \in \mc{B}_1 (\Hi)$ and $A\in \cB(\cH)$ we shall thus write either $\rho(A)$ or $\tr(\rho A)$ the expectation value of the observable $A$ in the state $\rho$. A state is faithful if $\forall A \in \mc{B}^+ (\Hi), A \neq 0 \Rightarrow \rho (A) > 0$ which is equivalent to $\rho\in \cB_1^+(\cH)$ being positive definite.

In the markovian description of open quantum systems, the evolution of a state is described by a completely positive and trace-preserving (CPTP) map (Schr\"odinger picture) while the one of an observable is described by a completely positive and unital map (Heisenberg picture), see Section \ref{open quantum systems}. A bounded linear map $\Phi$ acting on $\mc{B}(\Hi)$ is called positive (or positivity preserving) if for any $A \in \mc{B}^+  (\Hi)$, $\Phi (A)\in \mc{B}^+(\Hi)$, and it is called unital if $\Phi(\un) = \un$. $\Phi$ is said to be completely positive if for any $d > 0$, $\Phi \otimes \un_{\mc{B}(\mb{C}^d)}$ is a positive map. A linear map $\Phi$ acting on $\mc{B}_1 (\Hi)$ is called trace-preserving if $\forall \rho \in \mc{B}_1 (\Hi), \tr(\Phi(\rho)) = \tr(\rho)$. If $\Phi\in \mc{B}(\mc{B}_1(\Hi))$ then $\Phi^* \in \mc{B}(\mc{B}(\Hi))$ denotes its dual map, i.e. $\tr(\Phi(\rho) A)=\tr(\rho \, \Phi^*(A))$ for any $\rho\in \cB_1(\cH)$ and $A\in \cB(\cH)$. Note that $\Phi$ is trace-preserving iff $\Phi^*$ is unital. 

If $\cL$ is a CPTP map then $\ds (\cL^n)_{n\in\N}$ is called a quantum dynamical semi-group. Obviously the large time ($n\to\infty$) limit of quantum dynamical semi-groups is closely related to the spectral properties of the map $\cL$. It is known that positive maps satisfy Perron-Frobenius type results, see e.g. \cite{EHK}. In particular their spectral radius is always an eigenvalue and there is a non-negative corresponding eigenvector. If moreover $\cL$ is CPTP then its spectral radius is $1$ so that $\cL$ admits an invariant state (in finite dimension). The notion of primitive CPTP map will play an important in the paper. Recall that any completely positive map can be written in the form
\begin{equation}\label{eq:kraus}
\cL(\cdot) = \sum_{i\in I} V_i \, \cdot \, V_i^*,
\end{equation}
where $I$ is an at most countable set ($I$ is finite in finite dimension). Such a form is called a Kraus representation of $\cL$ \cite{Kr}.
\begin{defi}\label{def:primitive} Let $\cL$ be a completely positive map on $\cB(\cH)$ given by (\ref{eq:kraus}). Then $\cL$ is primitive if there exists $n\in \N$ such that ${\rm Span}\left\{V_{i_1}\ldots V_{i_n}\, |\, i_1,\ldots,i_n\in I \right\}=\cB(\cH)$.
\end{defi}
\begin{rema} There are actually several equivalent definitions of primitivity for positive maps. The one given here is the simplest for our purpose. We refer the reader to e.g. \cite{Wo} for a more detailed discussion on the subject.
\end{rema}

The importance of the notion of primitive map for our purpose is due to the following
\begin{prop}\label{spprop} Let $\cL$ be a completely positive map  and denote by $r$ its spectral radius. $\cL$ is primitive iff its spectral radius is a simple dominant eigenvalue, i.e. all other eigenvalues satisfy $|\lambda|<r$, with positive definite left and right eigenvectors.
\end{prop}

When $\cL$ is CPTP its spectral radius is $1$ and its dual map $\cL^*$ is unital, so that $\un$ is a left eigenvector. Hence the above proposition can simply be rephrased as ``a CPTP map $\cL$ is primitive if and only if $1$ is a simple dominant eigenvalue and $\cL$ admits a (unique) faithful invariant state''. The notion of primitive CPTP map is thus immediately related to strong ergodic properties of the corresponding quantum dynamical semi-group. 
\begin{prop}\label{prop:ergodicprimitive} Let $(\Lj^n)_{n\in\mb{N}}$ be a quantum dynamical semigroup. $\cL$ is primitive if and only if for any state $\rho$ one has 
$\ds \lim_{n \rightarrow +\infty} \Lj^n (\rho) = \rho_+,$  where $\rho_+$ is the unique faithful invariant state of $\cL$. In other words $\rho_+$ is mixing for the semi-group generated by $\cL$.
\end{prop}

A slightly stronger notion is sometimes useful. A map $\cL$ is called positivity improving if for any $X\in \cB^+(\cH)\setminus\{0\}$ its image $\Lj(X)$ is positive definite. It is easy to see that if a CP map is positivity improving then it is primitive. 

Finally, if $\Hi_1$, $\Hi_2$ are two finite dimensional Hilbert spaces and $\rho\in \cB_1(\cH_1\otimes \cH_2)$, we denote by $\rho_1 := \tr_{\Hi_2} (\rho)$ the partial trace of $\rho$ with respect to $\cH_2$. $\rho_1$ is the unique element in $\Bo_1 (\Hi_1)$ such that for all $A \in \Bo (\Hi_1)$ one has $\tr (\rho \times A \otimes \un_{\Hi_2}) = \tr (\rho_1 A)$.

Similarly fix $\rho_{2} \in \Bo_1 (\Hi_2)$ and let $A \in \Bo (\Hi_1 \otimes \Hi_2)$. Then  $A_1:=\tr_{\Hi_2} (\un_{\Hi_1} \otimes \rho_2 \times A)$ is the unique element in $\Bo (\Hi_1)$ such that for any $\rho_1 \in \Bo_1 (\Hi_1)$, $\tr (\rho_1 \otimes \rho_2 \times A) = \tr (\rho_1 A_1)$. $A_1$ is called the partial trace of $A$ w.r.t. the state $\rho_2$, and will be denoted $A_1 = \tr_{\rho_2} (A)$.


\subsection{Open quantum systems}\label{open quantum systems}

A quantum system $\cS$ is said to be open when it interacts with another quantum system $\En$. $\Sy$ is sometimes called the small system, $\En$ the environment (which can further be made of several components), and $\Sy + \En$ is the global system. If $\cH_\cS$ and $\cH_\cE$ are the Hilbert spaces of $\cS$ and $\cE$ respectively then the Hilbert space of the global system is $\Hi_{\Sy + \En}=\cH_\cS\otimes\cH_\cE$, and if the joint system $\Sy + \En$ is in the state $\rho_{\Sy + \En} \in \mc{B}_1 (\Hi_{\Sy} \otimes \Hi_{\En})$ then $\Sy$ is in the reduced state $\rho_\cS = \tr_{\Hi_{\En}} (\rho_{\Sy + \En})$. Similarly, if the system is initially decoupled, i.e. $\Sy + \En$  is in a state $\rho_{\Sy + \En} = \rho_\cS \otimes \rho_{\En}$ where $\rho_\cS$ and $\rho_{\En}$ are the states of $\cS$ and $\En$ respectively, then for an observable $A \in \mc{B} (\Hi_{\Sy} \otimes \Hi_{\En})$ of the global system, the observable ``seen by $\Sy$'' is $\tr_{\rho_{\En}} (A) = \tr_{\Hi_{\En}} (\un \otimes \rho_{\En} \times A)$.

The non-interacting dynamics of $\Sy$ and $\En$ are described by hamiltonians $H_\cS$ and $H_\cE$ respectively and the interaction between $\cS$ and $\En$ by some interaction operator $V$ acting on $\Hi_{\Sy} \otimes \Hi_{\En}$. Then $H:= H_{\Sy} \otimes \un + \un \otimes H_{\En} + V$ is the hamiltonian for the joint evolution. In the sequel we will often omit the inessential factors $\un$ in the tensor products.

Fix an initial (or reference) state $\rho_\cE$ of the environment and suppose the system is in some initially decoupled state $\rho \otimes \rho_{\En}$. Then the state of the global system $\Sy + \En$ after some time $\tau>0$ is $U \times \rho \otimes \rho_{\En} \times U^*$ where $U:=e^{-i\tau H}$. Hence, the small system $\cS$ is in the state
\begin{equation}\label{def:rdm}
\Lj (\rho) = \tr_{\Hi_{\En}} (U \times \rho \otimes \rho_{\En} \times U^*).
\end{equation}
One easily checks that $\cL$ defines a CPTP map on $\mc{B}_1^+ (\Hi_{\Sy})$. $\Lj$ is called the reduced dynamics map of $\Sy$ associated to the open quantum system $\cS+\cE$ and for the duration $\tau$. Its dual map describes the evolution of observables $A \in \mc{B} (\Hi_{\Sy})$ and is given by 
\begin{equation*}
\Lj^* (A) = \tr_{\rho_{\En}} (U^* \times A \otimes \un \times U).
\end{equation*}


\subsection{Entropy production in an open system}\label{entropy}

The entropy observable of a quantum system in a faithful state $\rho$ is defined as $S(\rho) := - \log \rho$. The von Neumann entropy of $\rho$ is then the expectation value of $S(\rho)$ in the state $\rho$, i.e.  $\Ent(\rho) := -\tr (\rho \log \rho) = \rho (S(\rho))$. The relative entropy of $\rho$ relatively to the faithful state $\nu$ is then
\begin{equation*}
\Ent(\rho \, |\,  \nu) := \tr (\rho(\log \rho - \log \nu)) = \rho(S(\rho) - S(\nu)).
\end{equation*}
It is well-known that $\Ent(\rho, \nu)\geq 0$ and $\Ent(\rho, \nu)=0$ if and only if $\rho = \nu$.

Consider an open system $\cS+\cE$ as in Section \ref{open quantum systems}, interacting for a duration $\tau>0$. Denote by $U$ the unitary evolution of the joint system and $\cL$ the reduced dynamics on $\cS$ as given by (\ref{def:rdm}). Assume moreover that the system $\En$ is initially at equilibrium, i.e. that $\ds \rho_{\En}:= \frac{e^{-\beta H_{\En}}}{\tr(e^{-\beta H_{\En}})}$ is a Gibbs state at some inverse temperature $\beta$. Then a simple calculation shows that the following entropy balance equation holds, see also \cite{HJPR},
\begin{equation}\label{ebe}
\Ent(\Lj(\rho)) - \Ent(\rho) = \Ent(U \times \rho \otimes \rho_{\En} \times U^*\, |\, \Lj (\rho) \otimes \rho_{\En}) - \beta\tr\big((\rho \otimes \rho_{\En}) (U^* H_{\En} U - H_{\En})\big).
\end{equation}
The second term of the right-hand side can be interpreted as the entropy flux coming from $\En$ (it is $\beta$ times the energy variation in $\cE$). This motivates the following
\begin{defi} We define the entropy production of $\Sy$ during its interaction with $\En$ by
\begin{equation}\label{def:entropyprod}
\sigma := \Ent(U \times \rho \otimes \rho_{\En} \times U^*\, |\, \Lj (\rho) \otimes \rho_{\En}).
\end{equation}
\end{defi}

The entropy production is meant to be always nonnegative and to be zero iff there is no energy (or ``heat'') flux between $\Sy$ and $\En$. Indeed one easily gets that, see e.g. \cite{JPlan}, 
\begin{equation*}\label{eq:entropyprod}
 \sigma = 0 \ \Longleftrightarrow \ \tr \big(\rho \otimes \rho_{\En} (U^* H_{\En} U - H_{\En})\big) =  \Ent(\Lj(\rho)) - \Ent(\rho) = 0.
\end{equation*}
Note also that $\sigma=0$ iff $U \times \rho \otimes \rho_{\En} \times U^* = \Lj (\rho) \otimes \rho_{\En}$ that is if the system is initially in the non-entangled state $\rho\otimes \rho_\cE$  the interaction leaves them non-entangled with $\cE$ unchanged.

\begin{defi} Let $U$ be a unitary operator on $\cH_\cS\otimes\cH_\cE$, $\rho\in\cB_1(\cH_\cS)$ and $\rho_\cE\in\cB_1(\cH_\cE)$. We say that the triple $(U, \rho, \rho_{\En})$ satisfies the Non-Entanglement condition if
\begin{equation}\label{def:NEtriple}
U \times \rho \otimes \rho_{\En} \times U^* = \cL(\rho) \otimes \rho_{\En},
\end{equation}
in other words if the entropy production of the associated open system vanishes.
\end{defi}

The importance of this Non-Entanglement condition will be made more transparent in Section \ref{ssec:NE}. We mention that this condition also appears in \cite{HJPR} in the context of adiabatic RIS. At this point we simply note that if $(U, \rho, \rho_{\En})$ satisfies the Non-Entanglement condition and $\rho$ is an invariant state for $\cL$ then $\rho \otimes \rho_{\En}$ is an invariant state for the interacting dynamics. As a consequence the entropy observable
\begin{equation}\label{eq:entropyconservation}
 S(\rho\otimes\rho_{\En})=-\log(\rho)\otimes\un-\un\otimes\log(\rho_{\En})
\end{equation}
is a constant of motion.


\section{Non-equilibrium Repeated Interaction Systems}\label{ris}


\subsection{Repeated Interaction Systems (RIS)}\label{ris1}

Repeated Interaction Systems form a specific class of open quantum systems in which the environment $\cE$ has the following structure $\cE=\En^1+\En^2+\cdots\En^n+\cdots$ where $(\cE^n)_n$ is a sequence of quantum subsystems with the associated Hilbert spaces $\Hi_{\En^n}$ and free hamiltonians $H_{\En^n}$. $\Sy$ will be called the \textit{small system} and the $\En^n$'s will be called the \textit{probes}. The Repeated Interaction dynamics consists in the joint evolution of $\Sy$ and $\En^1$ for a duration $\tau_1$, immediately followed by the joint evolution of $\Sy$ and $\En^2$ for a duration $\tau_2$, etc. For any $n,p\in\N^*$, $\En^n$ and $\En^p$ are disjoint and never interact directly. The interaction between $\cS$ and $\cE^n$ is described by the interaction operator $V_n$. Hence during the $n$-th interaction the coupled hamiltonian is
\begin{equation*}
 H_n:=H_{\Sy}\otimes\un+\un\otimes H_{\En^n}+V_n,
\end{equation*}
and the unitary propagator for the coupled dynamics describing the $N$ first interactions is thus given by $U_NU_{N-1}\cdots U_1$ where $\ds U_n:=\e^{-i\tau_nH_n}$, $n=1,\ldots,N$.
Note that by construction the various $H_{\cE^n}$ commute and that $[H_{\cE^n},H_p]=0$ whenever $n\neq p$.

If $\cS$ is initially in the state $\rho$ and the $n$-th probe is in the state $\rho_{\cE^n}$ which we assume to be invariant for the free dynamics of $\cE^n$, then the state of $\cS$ after $N$ interactions is given by
\begin{equation}\label{rishamiltonianform}
\rho_N = \tr_{\cH_{\cE^1}\otimes \cdots\otimes \cH_{\cE^N}} \left( U_N\cdots U_1 \ \rho \otimes \rho_{\cE^1} \otimes \cdots\otimes \rho_{\cE^N} \ U_1^* \cdots U_N^* \right).
\end{equation}
It is easy to see that the repeated interaction structure induces the following markovian behaviour, see e.g. \cite{BJM} for more details, 
\[
\rho_N = \tr_{\cH_{\cE^N}} \left( U_N\, \rho_{N-1} \otimes \rho_{\cE^N}\, U_N^*\right) =: \cL_N(\rho_{N-1}),
\]
so that if, for any $n$, the map
\begin{equation*}
\cL_n (\rho) := \tr_{\Hi_{\En^n}} (U_n \times \rho \otimes \rho_{\En^n} \times U_n^*),
\end{equation*}
denotes the reduced dynamics map associated to the interaction with the probe $\cE^n$, we have
\begin{equation}\label{rismarkovianform}
\rho_N = \cL_N\circ\cdots\circ\cL_1(\rho).
\end{equation}

In the simplest situation of repeated interaction systems the probes are copies of an identical system $\cE$, i.e. $\cH_{\cE^n}\equiv\cH_\cE$, $H_{\cE^n}\equiv H_\cE$, $V_n\equiv V$, 
$\tau_n\equiv\tau$, $\rho_{\cE^n}\equiv\rho_\cE$. Then the reduced dynamics map coincide, $\cL_n\equiv \cL$, and the evolution of the system $\cS$ is given by $\rho_N=\cL^N(\rho)$. We shall refer to this situation as an equilibrium situation, the various probes being considered as elements of a single reservoir. One may think of the initial states of the probes as Gibbs states at some common inverse temperature $\beta$.


\subsection{Non-equilibrium RIS}\label{sec:noneqris}

In this paper we are interested in understanding RIS in a non-equilibrium situation. The picture one should have in mind is that the system $\cS$ is coupled to several reservoirs $\cR_1,\ldots,\cR_M$ ($M\geq 2$) which are initially in thermal equilibrium but with possibly different temperatures. The various probes are then associated to one of these $M$ reservoirs.

More precisely, we fix a finite set of quantum systems $\En_j$, $j=1,\ldots,M$, with the associated Hilbert spaces $\Hi_{\En_j}$, free hamiltonians $H_{\En_j}$, interaction operators $V_j$, interaction times $\tau_j$ and initial states $\rho_{\En_j}$. Then $H_j$, $U_j$ and $\cL_j$ denote respectively the interacting hamiltonian, acting on $\cH_\cS\otimes \cH_{\cE_j}$, unitary propagator and reduced dynamics map as given by (\ref{def:rdm}). For each interaction, the probe will be a copy of one of these $M$ systems. 
\begin{rema}\label{rem:simplepicture} To have a simple picture in mind, the reader should think of all data identical, e.g. $\cH_{\cE_j}\equiv \cH_\cE$, $H_{\cE_j}\equiv H_\cE$,... except for the inital states $\rho_{\En_j}$ which are Gibbs states at possibly different temperatures $\beta_j^{-1}$.
\end{rema}
The sequence of probes is then described by a sequence $\bj:=(j_n)_{n\in\N^*}\in \{1,\ldots, M\}^{\N^*}$ where $j_n$ will describe of which of the $M$ systems $\cE_j$'s the probe $\cE^n$ is a copy, i.e. $\forall n\in\N^*$ one has $\En^n\equiv\En_{j_n}$. Thus $\Sy$ interacts with $M$ reservoirs $\cR_1,\ldots,\cR_M$ where $\cR_j$ denotes the ``union'' of all the probes corresponding to the index $j$, i.e.
\[
\cR_j := \bigplus_{n,\, j_n=j} \ \cE^n, 
\]
and if $\rho$ is the initial state of $\Sy$, after $n$ interactions $\Sy$ is in the state
\begin{equation*}
 \rho_n(\bj):=\Lj_{j_n}\circ\Lj_{j_{n-1}}\circ\cdots\circ\Lj_{j_1}(\rho).
\end{equation*}
Accordingly, for an observable $A\in\Bo(\Hi_{\Sy})$, the corresponding Heisenberg evolution is
\begin{equation*}
 A_n(\bj):=\Lj_{j_1}^*\circ\Lj_{j_2}^*\circ\cdots\circ\Lj_{j_n}^*(A).
\end{equation*}

In this paper, we shall consider two specific and rather natural situation of non-equilibrium RIS whose reduced dynamics can be linked to a discrete quantum dynamical semigroup.


\subsubsection{Cyclic case}

A non-equilibrium RIS is called cyclic when $\Sy$ interacts first with $\Res_1$, then $\Res_2$, ..., $\Res_M$, then $\Res_1$ again, etc. In other words the sequence $\bj$ describing the interactions is the $M$-periodic sequence $\bj^{cy}=(j_n^{cy})_n$ where for any $k\in\{1,\ldots,M\}$ and $n\in\N$ one has $j_{k+nM}^{cy}=k$. Consequently, the reduced dynamics of $\Sy$ over the cycles can be described by the discrete semigroup $\ds \big(\Lj_{cy}^n\big)_n$ where
\begin{equation}\label{def:cycle}
\Lj_{cy}:=\Lj_M\circ\Lj_{M-1}\circ\cdots\circ\Lj_1.
\end{equation}
Note that $\ds \Lj_{cy}(\rho) = \tr_{\Hi_{\En_{cy}}}(U_{cy}\times\rho\otimes\rho_{cy}\times U_{cy}^*),$ where $\En_{cy}:=\En_1+\cdots+\En_M$, $\Hi_{\En_{cy}}:=\Hi_{\En_1}\otimes\cdots\otimes\Hi_{\En_M}$, $U_{cy}:=U_M\cdots U_1$ and $\ds \rho_{cy}:=\rho_{\En_1}\otimes\cdots\otimes\rho_{\En_M}.$


\subsubsection{Random case}\label{rdcase}

A non-equilibrium RIS is called random when the order in which the small system $\Sy$ interacts with the $(\Res_j)_{1\leq j\leq M}$ is described by a random process, i.e. for each interaction, the subsystem $\cE_j$ of which $\cE^n$ is a copy will be chosen randomly from $\cE_1,\ldots,\cE_M$. The motivation to consider random RIS is related to the question of Time-Reversal Invariance, see Remark \ref{pwtr} in Section \ref{ssec:NE and tri}. In order to define such a model, we denote by $p$ the uniform probability measure on $\{1,\ldots,M\}$ and we denote by $\mc{P}$ the standard convoluted probability measure on $\{1,\ldots,M\}^{\mb{N}^*}$ associated to $p$. For any $\mc{T} \subset \{1,\ldots,M\}^{\mb{N}^*}$, $\mc{P}(\mc{T})$ is the probability for the Random RIS to interact successively with the reservoirs $\Res_{j_1}, \Res_{j_2}, \cdots, \Res_{j_n}, \cdots$ with $\bj=(j_n)_n\in\mc{T}$. We denote by $\E_{\mb{P}}$ the expectation value w.r.t. a measure $\mb{P}$. The choice of the measure $\mc{P}$ traduces the fact that the order in which the reservoirs $\cR_1,\ldots,\cR_M$ will interact with $\cS$ is chosen in an i.i.d. manner but on average $\cS$ will interact as much with each of them. This is the simplest and only case we shall consider here. A generalization of our results to a Markovian situation will be considered in the forthcoming paper \cite{BJP}.

Then one easily gets that $\ds \E_{\cP}(\Lj_{j_n}\circ\Lj_{j_{n-1}}\circ\cdots\circ\Lj_{j_1})=\Lj_{ra}^n$ where
\begin{equation*}
\Lj_{ra}:=\mb{E}_p(\Lj_j) = \frac{1}{M} \left( \cL_1+\cdots+\cL_M\right).
\end{equation*}
Thus, at least in expectation, the reduced dynamics of $\Sy$ can be described by the semigroup $(\Lj_{ra}^n)_{n\in\N}$. We say at least because of the following theorem which was proven in \cite{BJMrd} (it is a particular case of Theorem 1.3 in that paper)
\begin{theo}\label{rd} Suppose that there exists $1\leq j\leq M$ such that $\Lj_j$ is primitive. Then $\Lj_{ra}$ is primitive. Moreover, the Random RIS dynamics converges almost surely and in the ergodic mean to the unique invariant state of $\Lj_{ra}$. More precisely, there is a subset $\mc{T} \subset \{1,\ldots, M\}^{\mb{N}^*}$ such that $\cP (\mc{T}) = 1$ and for any $\bj\in\mc{T}$, any density matrix $\rho$, and any family of observables $\{1,\ldots, M\}\ni j\mapsto A(j)\in \cB(\cH_\cS)$ one has
\begin{equation}\label{eq:randomergodic}
 \lim_{N \rightarrow +\infty} \frac{1}{N} \sum_{n=1}^{N} \tr\left(\Lj_{j_{n-1}} \circ \cdots \circ \Lj_{j_1} (\rho) \times A(j_n)\right) = \tr\left(\rho_+^{ra} \times \E_p(A)\right),
\end{equation}
where $\rho_+^{ra}$ denotes the unique invariant state of $\Lj_{ra}$ and $\E_p(A)=\frac{1}{M} (A(1)+\cdots+A(M))$.
\end{theo}

\begin{rema} The observables $A(j)$ should be understood in the following sense. They represent the same physical quantity but their expression might vary depending on which probe $\cS$ is interacting with. They are called instantaneous observables in \cite{BJMrd}. A typical example is that of energy flux observables, see Section \ref{sec:fluxobs}. 

If in particular $A(j)\equiv A$, i.e. one considers a fixed observable, then (\ref{eq:randomergodic}) indeed traduces the ergodicity of the random RIS.
\end{rema}

\noindent {\bf Notation.} Throughout the paper when assumptions, results or identities are formulated for both the cyclic and random situation we will often use the symbol $\s$ which will stand for either $cy$ in the cyclic case or $ra$ in the random case.


\subsection{Assumptions}

In this section we formulate the various assumptions that will be used in the paper.


\subsubsection{Temperatures and thermal forces}\label{temperatures}

The first assumption concerns the initial states of the probes. We will assume that they are initially in thermal equilibrium. Namely,
\begin{hypo}[{\bf KMS}]\label{KMS} For any $1\leq j\leq M$, $\ds \rho_{\En_j}=\frac{\e^{-\beta_j H_{\En_j}}}{\tr\left(\e^{-\beta_j H_{\En_j}}\right)}$ for some $\beta_j$.
\end{hypo}
\noindent From now on, we will always suppose Assumption (\nameref{KMS}) holds. 

Fix some $\beta_\refer > 0$. If $\zeta_j := \beta_{ref} - \beta_j$ then $\bfzeta = (\zeta_1, \zeta_2, \cdots, \zeta_M)$ denotes the vector of thermal forces. 
In the sequel, for any quantity $A$ depending on $\bfzeta$, we will write $A_\bfzeta$ when we need to stress its dependence on $\bfzeta$, and only $A$ otherwise. When all the temperatures are equal (but not necessarily to $\beta_\refer^{-1}$), i.e. $\bfzeta=(\zeta,\ldots,\zeta)$, we will write $A_\zeta$ instead of $A_\bfzeta$. 

Note that $\cE_j$ depends only on $\zeta_j$. Hence we may write $\cL_{j,\bfzeta}= \cL_{j,\zeta_j}$, $\rho_{\cE_j,\bfzeta}=\rho_{\cE_j,\zeta_j}$,...


\subsubsection{Ergodicity}\label{sec:ergodicity}

The next assumption concerns spectral properties of the reduced dynamics maps $\Lj_{\s,\bfzeta}$, $\s=cy$ or $ra$, and the asymptotic behaviour of $\Sy$ (recall Proposition \ref{prop:ergodicprimitive}).
\begin{hypo}[{\bf ER$\s$}]\label{ergodicity}
 There exists $\bfzeta\in\mb{R}^M$ such that $\Lj_{\s,\bfzeta}$ is primitive.
\end{hypo}

\begin{prop}\label{erprop} Suppose assumption (\nameref{ergodicity}) holds. Then
\begin{enumerate}
\item For any $\bfzeta\in\R^M$, $\Lj_{\s,\bfzeta}$ is primitive. As a consequence, for any $\bfzeta\in\R^M$ and $\rho\in\Bo_1(\Hi_{\Sy})$, $\ds \lim_{n\to\infty} \Lj_{\s,\bfzeta}^n(\rho)= \tr(\rho)\rho_{+,\bfzeta}^\s$  where $\rho_{+,\bfzeta}^\s$ denotes the unique invariant faithful state of $\Lj_{\s,\bfzeta}$.
\item $\rho_{+,\bfzeta}^\s$ is infinitely differentiable w.r.t. $\bfzeta \in \mb{R}^M$.
\end{enumerate}
\end{prop}

\begin{proof} 1. Let $\lambda_{jk}$, $\varphi_{jk}$ denote the eigenvalues and eigenvectors of $H_{\cE_j}$, i.e. $\ds H_{\cE_j}=\sum_k \lambda_{jk} |\varphi_{jk}\rangle\langle \varphi_{jk}|$. Then it follows from (\ref{def:rdm}) that 
\[
\cL_{j,\zeta_j}(\rho) = \sum_{k,k'} \frac{\e^{-\lambda_{jk}(\beta_\refer-\zeta_j)}}{Z_{\zeta_j}} V_{jkk'} \, \rho\, V_{jkk'}^*, 
\]
where $\ds V_{jkk'} = \left\langle \varphi_{jk'}, \e^{-i\tau_jH_j} \varphi_{jk}\right\rangle$ does not depend on $\bfzeta$. It follows immediately that the maps $\Lj_{\s,\bfzeta}$ have Kraus decomposition of the form
$\ds
\Lj_{\s,\bfzeta} = \sum_i f_i(\bfzeta) V_i \rho V_i^*
$ 
where the $V_i$ do not depend on $\bfzeta$ and with $f_i(\bfzeta)>0$.
For any $\bfzeta,\bfzeta'$ the maps $\Lj_{\s,\bfzeta}$ and $\Lj_{\s,\bfzeta'}$ therefore have the same Kraus decomposition up to positive scalar factors which proves 1 (recall Definition \ref{def:primitive}). 

2. It follows from 1. and Proposition \ref{spprop} that $1$ is an isolated simple eigenvalue for any $\bfzeta$ so the result follows by regular perturbation theory \cite{Kato}.
\end{proof}

Finally, we have the following result which shows that it suffices to have information on one of the probes to get information on the entire non-equilibrium RIS.
\begin{prop} Let $\bfzeta\in\R^M$. If $\cL_{j,\zeta_j}$ is primitive for some $1\leq j\leq M$  then (\nameref{ergodicity}) holds for $\s=ra$, and if moreover $\cL_{j,\zeta_j}$ is positivity improving then (\nameref{ergodicity}) holds for $\s=cy$.
\end{prop}
The proof follows directly from Lemma \ref{positive definite} below in the cyclic case and from Theorem \ref{rd} in the random case. 

\begin{lemm}\label{positive definite} If $\rho$ is positive definite then so is $\Lj_j(\rho)$. As a consequence, if there exists $1\leq j\leq M$ such that $\cL_{j,\zeta_j}$ is positivity improving, then so is $\cL_{cy,\bfzeta}$.
\end{lemm}

\begin{proof} If $\rho$ is positive definite so is $\rho_{\cE_j}$ and hence $\rho\otimes \rho_{\cE_j}$. Thus for any non-zero $A\in \cB^+(\cH_\cS)$ one has
\[
\tr\left( \cL_j(\rho)\, A\right) = \tr\left( \rho\otimes \rho_{\cE_j} \times U_j^* A\otimes \un \,U_j\right) >0,
\]
so that $\cL_j(\rho)$ is indeed positive definite.
\end{proof}


\subsubsection{Time-Reversal}\label{ssec:tri}

The next assumption concerns Time-Reversal Invariance. A time reversal of a quantum system $(\cH,H)$ is an antiunitary involution $\theta$ on $\cH$ such that $\theta H=H\theta$. Given such a $\theta$ we denote by $\Theta$ the antilinear $*$-automorphism acting on $\cB(\cH)$ as $\Theta(X)=\theta X\theta$. Note in particular that one has $\Theta (\e^{-itH}) =\e^{itH}$ for all $t\in \R$, or in an equivalent way if $\tau_t(X):=\e^{itH}X\e^{-itH}$ is the Heisenberg evolution of an observable $X$ then 
 \begin{equation}\label{eq:tri}
 \Theta \circ \tau_t \circ \Theta =\tau_{-t}, \quad \forall t\in\R.
 \end{equation}
In open systems one usually further specifies the structure of $\theta$. Namely we assume that we are given time reversal $\theta_\cS$ and $\theta_\cE$ associated to $(\cH_\cS,H_\cS)$ and $(\cH_\cE,H_\cE)$ respectively and such that $\theta_\cS \otimes \theta_\cE V = V \theta_\cS \otimes \theta_{\En}$ where $V$ is the interacting operator. Then $\theta=\theta_\cS\otimes\theta_\cE$ is a time reversal for the coupled system $(\cH_\cS\otimes \cH_\cE,H_\cS+H_\cE+V)$.

A state $\rho$ is then called time-reversal invariant for $\theta$ if $\Theta(\rho)=\rho$ which is equivalent to $\rho(\Theta(X))=\rho(X^*)$ for all $X\in\cB(\cH)$, and in particular $\rho(\Theta(X))=\rho(X)$ when $X$ is an observable hence self-adjoint. Finally, a quantum system $(\cH,H,\rho)$ is called time-reversal invariant iff there exists a time reversal $\theta$ on $(\cH,H)$ such that $\rho$ is time-reversal invariant for $\theta$. We consider the following assumption.

\begin{hypo}[{\bf TRI}]\label{Time-Reversal}
There exist antiunitary involutions $\theta$ and $(\theta_{\En_j})_{1\leq j\leq M}$ acting on $\Hi_{\Sy}$ and $(\Hi_{\En_j})_{1\leq j\leq M}$ such that $\theta H_{\Sy} = H_{\Sy} \theta$, $\theta_{\En_j} H_{\En_j} = H_{\En_j} \theta_{\En_j}$ and $\theta \otimes \theta_{\En_j} V_j = V_j\theta\otimes\theta_{\En_j}$. 
\end{hypo}
\noindent Since the $\rho_{\cE_j}$ are KMS states (\nameref{Time-Reversal}) guarantees that all the probes $(\cH_{\cE_j},H_{\cE_j},\rho_{\cE_j})$ are time-reversal invariant systems. We shall come back to time-reversal invariance in Section \ref{ssec:NE and tri}.


\subsection{The Non-Entanglement condition}\label{ssec:NE}

There is a last assumption, see Assumption (\nameref{Non-Entanglement}) p.\pageref{Non-Entanglement}, which will play an important role in our paper and which is very specific to RIS. The goal if this section is first to explain its origin and then derive some of its consequences.


\subsubsection{The Non-Entanglement condition: a signature of equilibrium}

One of our goal is to understand linear response theory for RIS, which means how does the system respond to a small perturbation \emph{from equilibrium}. It is therefore important to specify what do we mean by equilibrium. In this perspective the most ideal situation is certainly that of a single type of probes, i.e. $M=1$, all initially at the same inverse temperature $\beta$. One of the usual feature of equilibrium is that it is characterized by the vanishing of entropy production. During its interaction with a single probe $\cE$ the entropy production of the system is given by (\ref{def:entropyprod}), and if the system is initially in an invariant state $\rho$ we simply have 
\[
\sigma = \Ent(U \times \rho \otimes \rho_{\En} \times U^*\, |\, \rho \otimes \rho_{\En})
\]
which vanishes iff $U \times \rho \otimes \rho_{\En} \times U^* = \rho \otimes \rho_{\En}$. It is thus natural to require that for each species of probes, and whatever is their initial temperature, there exists an invariant state $\rho_j$ such that the triple $(U_j,\rho_j,\rho_{\cE_j})$ satisfies the non-entanglement condition (\ref{def:NEtriple}). 

Consider now the general framework of Section \ref{sec:noneqris} and more precisely the random situation of Section \ref{rdcase}. When Assumption (\nameref{ergodicity}) holds the entropy production in the system is given by, see Proposition \ref{prop:2ndlaw} and also \cite{BJMrd,BJM}, 
\begin{equation}\label{eq:randomentropyprod}
\sigma_{ra}^+= -\sum_{j=1}^M\beta_j\rho_+^{ra}(\Phi_j),
\end{equation}
where $\ds \Phi_j = -\frac{1}{T} \tr_{\rho_{\En_j}} (U_j^* H_{\En_j} U_j - H_{\En_j})$, $T=\tau_1+\cdots+\tau_M$, denotes the energy flux observable associated to the $j$-th type of probe (see Section \ref{sec:fluxobs} for more details about these $\Phi_j$). A natural notion of equilibrium is that it leads to a vanishing of entropy production. The following proposition gives a simple characterization in terms of the individual probes.
\begin{prop}\label{prop:NEandequilibrium} If Assumption (\nameref{ergodicity}) holds, then $\sigma_{ra}^+=0$ if and only if  
\begin{equation}\label{eq:nonentanglement}
U_j\times\rho_+^{ra} \otimes\rho_{\En_j}\times U_j^*  =\rho_+^{ra}\otimes\rho_{\En_j}, \quad \forall j\in\{1,\ldots,M\},
\end{equation}
i.e. the states $\rho_{+}^{ra} \otimes\rho_{\En_j}$ are invariant states of the joint systems $\cS+\cE_j$ so in particular $\cL_j(\rho_+^{ra})=\rho_+^{ra}$ for all $j$. In other words $\rho_+^{ra}$ is a common invariant state for all the probes and the triples $(U_j,\rho_+^{ra},\rho_{\cE_j})$ all satisfy the Non-Entanglement condition (\ref{def:NEtriple}).
\end{prop}

\begin{proof} One direction is obvious. Namely, if (\ref{eq:nonentanglement}) holds for all $j$ one easily computes 
\begin{equation*}\label{eq:individualfluxes}
\rho_+^{ra}(\Phi_j) = -\frac{1}{T} \tr\left[ \rho_+^{ra} \otimes \rho_{\En_j} (U_j^* H_{\En_j} U_j - H_{\En_j})\right] = 0, 
\end{equation*}
i.e. all the steady fluxes vanish, hence $\ds \sigma_{ra}^+= -\sum_{j=1}^M\beta_j\rho_+^{ra}(\Phi_j) = 0$.

\noindent Suppose now that $\sigma_{ra}^+=0$. Using (\ref{ebe}) we have for all $j\in\{1,\ldots,M\}$
\begin{eqnarray*}
\lefteqn{ \Ent(\cL_j(\rho_+^{ra})) - \Ent(\rho_+^{ra}) } \\
 & = & \Ent\Big(U_j \times \rho_+^{ra} \otimes \rho_{\En_j} \times U_j^*\, |\, \cL_j (\rho_+^{ra}) \otimes \rho_{\En_j}\Big) - \beta_j\tr\Big((\rho_+^{ra} \otimes \rho_{\En_j}) (U_j^* H_{\En_j} U_j - H_{\En_j})\Big)\\
 & = & \Ent\Big(U_j \times \rho_+^{ra} \otimes \rho_{\En_j} \times U_j^*\, |\, \cL_j (\rho_+^{ra}) \otimes \rho_{\En_j}\Big) + \beta_j T \rho_+^{ra}(\Phi_j).
\end{eqnarray*}
Summing these identities over $j$ and using (\ref{eq:randomentropyprod}) we get
\begin{equation}\label{ebeNE}
\frac{1}{M}\sum_{j=1}^M \Ent\big(\cL_j(\rho_+^{ra})\big) - \Ent(\rho_+^{ra}) = \frac{1}{M}\sum_{j=1}^M \Ent\Big(U_j \times \rho_+^{ra} \otimes \rho_{\En_j} \times U_j^*\, |\, \cL_j (\rho_+^{ra}) \otimes \rho_{\En_j}\Big).
\end{equation}
Relative entropies are non-negative quantities so that the right-hand side is non-negative. But the left-hand side is non-positive. Indeed von Neumann entropy is strictly concave so that, using $\cL_{ra} (\rho_+^{ra})=\rho_+^{ra}$, we have
\begin{eqnarray*}
\lefteqn{ \frac{1}{M}\sum_{j=1}^M \Ent\big(\cL_j(\rho_+^{ra})\big) - \Ent(\rho_+^{ra}) } \\
 & \leq & \Ent\left( \frac{1}{M}\sum_{j=1}^M \cL_j(\rho_+^{ra}) \right) - \Ent(\rho_+^{ra}) = \Ent\left(\cL_{ra} (\rho_+^{ra})\right) -\Ent(\rho_+^{ra}) =0,
\end{eqnarray*}
with equality iff all the $\cL_j(\rho_+^{ra})$ are equal. This proves that both sides of (\ref{ebeNE}) vanish. 

As a consequence all the $\cL_j(\rho_+^{ra})$ are equal and, because $\cL_{ra} (\rho_+^{ra})=\rho_+^{ra}$, they have to be equal to $\rho_+^{ra}$, i.e. $\rho_+^{ra}$ is a common invariant state for all the probes. Moreover since all the terms on the right-hand side of (\ref{ebeNE}) are non-negative they are all zero, i.e.  
\[
\Ent\big(U_j \times \rho_+^{ra} \otimes \rho_{\En_j} \times U_j^*\, |\, \cL_j (\rho_+^{ra}) \otimes \rho_{\En_j}\big) =0 \quad \Rightarrow \quad U_j \times \rho_+^{ra} \otimes \rho_{\En_j} \times U_j^*= \cL_j (\rho_+^{ra}) \otimes \rho_{\En_j},
\]
that is the triples $(U_j,\rho_+^{ra},\rho_{\cE_j})$ satisfy the Non-entanglement condition (\ref{def:NEtriple}). 
\end{proof}

The above proposition refers only to the vanishing of entropy production. If we want to consider the case where \emph{all the temperatures are equal} as an \emph{equilibrium situation}, whatever is this temperature, this immediately leads to the following 
\begin{hypo}[{\bf NE}]\label{Non-Entanglement} There exists a function $\R\ni \zeta\mapsto \rho_{+,\zeta}$ such that, for any $j\in\{1,\ldots,M\}$ and $\zeta\in\R$, the triple $(U_j,\rho_{+,\zeta},\rho_{\En_j,\zeta})$ satisfies the following Non-Entanglement condition
\begin{equation*}\label{def:nonentanglement}
U_j\times\rho_{+,\zeta} \otimes\rho_{\En_j,\zeta}\times U_j^*  =\rho_{+,\zeta}\otimes\rho_{\En_j,\zeta}.
\end{equation*}
\end{hypo}
Assumption (\nameref{Non-Entanglement}) may look quite restrictive at first sight. Indeed it requires that any probe at any temperature has a non-entangled invariant state, and that this invariant state depends only on the temperature and not on the probes themselves. Proposition \ref{prop:NEandequilibrium} however shows that this is the natural condition if one wants to consider equal temperatures as an equilibrium situation, in the sense of vanishing of entropy production.

\begin{rema}\label{rem:NEequilibriumstate} If moreover Assumption (\nameref{ergodicity}) holds this obviously implies that, \emph{when all the temperatures are equal}, the unique invariant states in the cyclic and random case coincide with $\rho_{+,\zeta}$, i.e. $\rho^{cy}_{+,\zeta}=\rho^{ra}_{+,\zeta}=: \rho_{+,\zeta}$.
\end{rema}

In the rest of this section we derive two consequences of this non-entanglement assumption which will play an important role in our analysis.


\subsubsection{Effective hamiltonian}\label{ssec:effectiveham}

The first consequence is the existence of a conserved quantity for the $\cL_j$'s which plays the role of an effective hamiltonian. This is directly related to the conservation of the entropy observable (\ref{eq:entropyconservation}), see the end of Section \ref{entropy}.

\begin{prop}\label{prop:effectham} Suppose (\nameref{ergodicity}) and (\nameref{Non-Entanglement}) hold and let $H_\cS'=-\log \left( \rho_{+,0}^{1/\beta_\refer} \right)$ so that $\rho_{+,0}= \e^{-\beta_\refer H_\cS'}$. Then for any $\zeta \in \R$ the state $\rho_{+,\zeta}$ is a Gibbs state at inverse temperature $\beta=\beta_\refer-\zeta$ for the effective hamiltonian $H_{\Sy}'$, i.e. $\ds \rho_{+,\zeta} =\frac{\e^{-\beta H_{\Sy}'}}{\tr\left(\e^{-\beta H_{\Sy}'}\right)}$, $\beta=\beta_\refer-\zeta.$

Moreover, for any $j$, the observable $H_{\Sy}' + H_{\En_j}$ is a conserved quantity of the interacting dynamics, namely
\begin{equation}\label{conservation}
 U_j^*\left(H_{\Sy}'\otimes\un+\un\otimes H_{\En_j}\right)U_j=H_{\Sy}'\otimes\un+\un\otimes H_{\En_j}.
\end{equation}
\end{prop}

\begin{proof} Let $\zeta\in\R$ and $\beta=\beta_\refer-\zeta$. Assumption (\nameref{ergodicity}) together with Proposition \ref{erprop} imply that $\rho_{+,0}$ is positive definite hence $H_\cS'$ is well defined and Assumption (\nameref{Non-Entanglement}) guarantees that, for any $1\leq k\leq M$, $\left[\rho_{+,0} \otimes \e^{-\beta_\refer H_{\En_k}}, U_k \right] = 0$. So
\begin{equation}\label{eq:effecthamproof}
\left[\rho_{+,0}^{\frac{\beta}{\beta_\refer}} \otimes \e^{-\beta H_{\En_k}}, U_k \right] =
\left[\e^{-\beta H_{\Sy}'} \otimes \e^{-\beta H_{\En_k}}, U_k \right] = 0, \quad \forall \beta.
\end{equation}
Hence $\ds \frac{\e^{-\beta H_{\Sy}'}}{\tr\left( \e^{-\beta H_{\Sy}'}\right)}$ is an invariant state of $\cL_{k,\zeta}$ for all $k$ and hence of $\Lj_{\s,\zeta}$. Since the latter admits only one invariant state, it coincides with $\ds\rho_{+,\zeta}$.

Finally, (\ref{conservation}) is a direct consequence of (\ref{eq:effecthamproof}).
\end{proof}

\begin{rema} Note that the effective hamiltonian $H_\cS'$ is intrinsic to the system and that changing the reference temperature $\beta_\refer$ only amounts to an irrelevant shift by a constant.
\end{rema}

We also note the following lemma which can be seen as a sort of gauge invariance and whose proof is a straightforward computation left to the reader.
\begin{lemm}\label{lem:gauge} If Assumptions (\nameref{ergodicity}) and (\nameref{Non-Entanglement}) hold and $H_\cS'$ is as in Proposition \ref{prop:effectham}, then
\[
\cL_j\left( \e^{-itH_\cS'}\, \rho\, \e^{itH_\cS'}\right) = \e^{-itH_\cS'}\cL_j(\rho)\, \e^{itH_\cS'}, \quad \forall j\in\{1,\ldots,M\}, t\in\R, \rho\in\cB_1(\cH_\cS).
\]
\end{lemm}

\begin{rema}\label{rem:effecthaminvariance} Of course one has the same property in the Heisenberg picture,
\begin{equation}\label{eq:gauge}
\cL_j^*\left( \e^{itH_\cS'}\, X \, \e^{-itH_\cS'}\right) = \e^{itH_\cS'}\cL_j^*(X)\e^{itH_\cS'}.
\end{equation}
In particular, taking $X=H_\cS'$, we get that $[\cL_j^*(H_\cS'),H_\cS']=0$ for all $j$. 
\end{rema}


\subsubsection{Time-reversal invariance}\label{ssec:NE and tri}

The second consequence is related to time-reversal invariance. Assumption (\nameref{Time-Reversal}) is written in terms of the full interacting dynamics $U_j$. 
However the central objects in RIS are the reduced dynamics maps $\cL_j$. It is therefore natural, and important, to understand what are the consequences of Assumption (\nameref{Time-Reversal}) on them. 

Let us first recall the notion of time-reversal invariance for open systems whose evolution is given by a semi-group $(\cL^n)_n$ where $\cL$ is a given CPTP map. It is not a priori clear what should be the analog of (\ref{eq:tri}). The following definition is due to \cite{JPW}. It can be traced back up to \cite{Ag}, see also e.g. \cite{FU,Ma}. 
\begin{defi}\label{def:trimarkovsystem} A pair $(\cL,\rho)$, where $\cL$ is a CPTP map acting on $\cB_1(\cH)$ and $\rho$ a state, is called time-reversal invariant if there exists an antilinear $*$-automorphism $\Theta$ such that $\Theta(\rho)=\rho$ and $\Theta\circ \cL^*\circ \Theta = \cL^{*\rho}$ where $\cL^{*\rho}$ denotes the $\rho$ adjoint of $\cL^*$, i.e. its adjoint with respect to the inner product $\langle A,B\rangle_{\rho} := \tr\left( \rho A^*B\right)$.
\end{defi}

\begin{rema}\label{rem:origintrimarkov} One way to understand the above definition is as follows, see also \cite{JPW} for the case where the open system is derived via a weak coupling limit procedure \`a la van Hove. Typically the state $\rho$ is chosen to be an invariant state of the dynamics. If $(\cH,H,\rho)$ is a quantum system with $\rho$ an invariant state, it is easy to see that for all $t\in\R$ the $\rho$ adjoint $\tau_t^\rho$ of the map $\tau_t$ is $\tau_{-t}$ so that (\ref{eq:tri}) amounts to $\Theta \circ \tau_t \circ \Theta = \tau_t^\rho$. In the markovian description of open systems one then simply replaces the unitary evolution $\tau_t$ by the markovian one $\cL^*$.
\end{rema}

The following proposition shows that the non-entanglement condition allows one to make the connection between our Assumption (\nameref{Time-Reversal}) and  Definition \ref{def:trimarkovsystem}.
\begin{prop}\label{lem:tri} Let $\theta$ and $\theta_\cE$ be time reversal for the quantum systems $(\cH_\cS,H_\cS)$ and $(\cH_\cE,H_\cE)$ respectively and $V\in\cB(\cH_\cS\otimes\cH_\cE)$ such that $\theta_\cS \otimes \theta_\cE V = V \theta_\cS\otimes\theta_\cE$. Let $\tau>0$ and $\rho_\cE$ be such that $(\cH_\cE,H_\cE,\rho_\cE)$ is time-reversal invariant, and let $\cL$ be given by (\ref{def:rdm}). If $\cL$ is primitive and its (unique) invariant state $\rho$ is such that the triple $(U,\rho,\rho_\cE)$, $\ds U=\e^{-i\tau H}$ with $H=H_\cS+H_\cE+V$, satisfies the non-entanglement condition (\ref{def:NEtriple}), then the pair $(\cL,\rho)$ is time-reversal invariant in the sense of Definition \ref{def:trimarkovsystem}.
\end{prop}

\noindent In our framework of RI systems we thus immediately get

\begin{coro}\label{coro:tri} If Assumptions (\nameref{ergodicity}), (\nameref{Time-Reversal}) and (\nameref{Non-Entanglement}) hold then for any $\zeta\in \R$ and any $j$ the pair $\left( \cL_{j,\zeta},\rho_{+,\zeta}\right)$ is time-reversal invariant in the sense that $\Theta(\rho_{+,\zeta})=\rho_{+,\zeta}$ and
\begin{equation}\label{eq:revLj}
\Theta\circ \cL_{j,\zeta}^* \circ \Theta = \cL_{j,\zeta}^{*\rho_{+,\zeta}}.
\end{equation}
\end{coro}

\begin{proof} For $X\in\cB(\cH_{\cE})$ let $\Theta_\cE(X):= \theta_\cE X \theta_\cE$. One immediately gets $\Theta\otimes\Theta_\cE (U)=U^*$ so that
\[
U\, \Theta(\rho) \otimes \rho_\cE \, U^*  =  U \, (\Theta \otimes \Theta_\cE)\big(\rho \otimes \rho_\cE \big) \,U^* 
  =   (\Theta \otimes \Theta_\cE) \Big(U^*\, \rho \otimes \rho_\cE \, U \Big) 
  =  (\Theta \otimes \Theta_\cE)\big(\rho \otimes \rho_\cE \big)
  =  \Theta (\rho\big) \otimes \rho_\cE,
\]
where we have used successively $\Theta_\cE(\rho_{\cE})=\rho_{\cE}$, $\Theta\otimes\Theta_\cE (U)=U^*$, (\ref{def:NEtriple})  and $\Theta_\cE(\rho_{\cE})=\rho_{\cE}$ again.
Hence $\Theta (\rho)$ is an invariant state for $\cL$. Since $\cL$ is primitive this proves that  $\Theta (\rho)=\rho.$

On the other hand, for any $A,B\in\cB(\cH_\cS)$ we have
\begin{eqnarray*}
\tr\left( \rho\, \cL^{*\rho}(A)^* B \right) & = & \tr\left( \rho\, A^* \cL^*( B) \right) \\
 & = & \tr_{\cS+\cE} \left( \rho A^*\otimes \rho_{\cE}  \times U^*\, B\otimes \un \, U\right)\\
 & = & \tr_{\cS+\cE} \left( B\rho\otimes \rho_{\cE}  \times U\, A^*\otimes \un \, U^*\right)\\
 & = & \overline{\tr_{\cS+\cE} \left( \Theta(B\rho)\otimes \Theta_\cE(\rho_{\cE})  \times U^*\, \Theta(A^*)\otimes \un \, U\right)}\\
 & = & \overline{\tr\left( \Theta(B\rho) \cL^*\left(\Theta(A)^*\right)\right)} \\
 & = & \tr\left( B\rho \left(\Theta\circ\cL^*\circ\Theta(A)\right)^*\right),
\end{eqnarray*}
which proves that $\Theta\circ\cL^*\circ\Theta=\cL^{*\rho}$. Here  we have used the definition of $\cL^{*\rho}$ in line $1$, of $\cL^*$ in line $2$, cyclicity of the trace and (\ref{def:NEtriple})  in line 3,  antilinearity of $\Theta/\Theta_\cE$ in line $4$, $\Theta_\cE(\rho_{\cE})=\rho_{\cE}$ and definition of $\cL^*$ in line 5, and finally that $\cL^*$ is completely positive (so that $\cL^*(X^*)=(\cL^*(X))^*$) in line $6$.
\end{proof}

\begin{rema}\label{pwtr} Note that for cyclic RIS Assumption (\nameref{Time-Reversal}) makes the RIS time-reversal invariant only for the duration of the joint evolution with each individual probe, that is only locally in time. Indeed, the time-reversal operator $\Theta$ does not change the order of the interactions. This lack of global time-reversal invariance will however be resolved in the random model.
\end{rema}

\begin{rema}\label{rem:effectivehamreversal} If (\nameref{ergodicity}), (\nameref{Time-Reversal}) and (\nameref{Non-Entanglement}) hold it also follows from (\ref{conservation}) that the effective hamiltonian $H_\cS'$ is invariant under time reversal, i.e. $\Theta(H_\cS')=H_\cS'$.
\end{rema}


\subsection{A toy example}\label{toy model}

Suppose that $\cH_\cS$ and $\cH_\cE$ are copies of $\C^2$, $H_\cS=E a^* a$, $H_\cE = E_0 b^* b$, and the interaction $V$ has the form $V=\frac{\lambda}{2} (a^* \otimes b + a \otimes b^*)$. Here $a/a^*$, $b/b^*$ are the usual annihilation/creation operators on $\Hi_{\Sy}$ and $\Hi_{\En}$ respectively, i.e. $a=b=\begin{pmatrix} 0 & 1 \\ 0 & 0 \end{pmatrix}$. $E$ and $E_0$ denote the excited energy levels of $\cS$ and $\cE$ respectively, and the interaction consists in an exchange of excitation between $\Sy$ and $\En$. This model can be seen as a toy version of the Jaynes-Cummings hamiltonian describing the interaction between one mode of a quantized electro-magnetic field in a cavity and a two-level atom, see e.g. \cite{CDG}, and the corresponding RIS as a toy version of the one-atom maser model studied in \cite{BP,Bru14}.

One easily sees that the total number operator $N_{\Sy} \otimes \un + \un \otimes N_{\En}\equiv a^*a\otimes\un +\un \otimes b^*b$ is a conserved quantity. As a consequence, if $\rho_{\En} = \frac{\e^{-\beta H_{\En}}}{\tr\left(\e^{-\beta H_{\En}}\right)}$, $H_\cE=E_0 b^*b$, then $\rho = \frac{\e^{-\beta H_\Sy'}}{\tr\left(\e^{-\beta H_\Sy}\right)}$, $H_\cS'= E_0 N_\cS$, is an invariant state and the triple $\left(\e^{-i\tau H},\rho,\rho_{\En}\right)$ satisfies condition (\ref{def:NEtriple}).

Consider now the RIS where the probes $\cE_j$ are copies of $\cE$ with possibly different temperatures. We have just seen that Assumption (\nameref{Non-Entanglement}) indeed holds and leads to the effective hamiltonian $H_\cS'= E_0 N_\cS$. Assumption (\nameref{Time-Reversal}) is then obviously satisfied with $\theta$ and $\theta_{\En_j}$ the complex conjugation operations in the canonical bases of $\C^2$.
 
Finally let $\nu_j := \sqrt{(E - E_0)^2 + \lambda_j^2}$. A simple calculation shows that if $\nu_j\tau_j$ is not a multiple of $2\pi$ then $\cL_j$ is primitive and that its unique invariant state is $\rho_+(\zeta_j):=\frac{e^{-\beta_j H_\cS'}}{\tr\left(e^{-\beta_j H_\cS'}\right)}$ (see e.g. \cite{BJM} for more details). One thus infers that, if at least one of the $\nu_j\tau_j$'s is not a multiple of $2\pi$, at equilibrium the maps $\cL_{cy}$ and $\cL_{ra}$ are primitive as well so that (\nameref{ergodicity}) holds.


\section{Linear response of energy fluxes and entropy production}\label{lrtentropy}

\subsection{Energy flux observables}\label{sec:fluxobs}

The energy flux observables describe the energy fluxes that get out of the reservoirs $\Res_j$, as they are seen by the small system $\Sy$. Moreover we have to take into account the discrete-time nature of the RIS dynamics to define these fluxes. In other words we choose to study averaged fluxes, averaged over the duration of one interaction, instead of instantaneous ones. 

Clearly the reservoir $\cR_j$ can exchange energy only when it interacts with $\cS$. Moreover the typical time scale of the non-equilibrium RIS is $T=\tau_1+\cdots+\tau_M$ where we recall that $\tau_j$ is the duration of the interaction with probes $\cE_j$. This leads to the following 
\begin{defi}\label{def:flux} The energy flux observable associated to $\cR_j$ is
\begin{equation*}
\Phi_j = -\frac{1}{T} \tr_{\rho_{\En_j}} (U_j^* H_{\En_j} U_j - H_{\En_j}).
\end{equation*}
\end{defi}

In the cyclic framework it will be convenient to also use a slightly different, but closely related, flux observable. If $\cS$ is in the state $\rho$, e.g. the steady state $\rho_{+,\bfzeta}$ (see Remark \ref{rem:steadyflux}), at the beginning of a cycle, i.e. before it interacts with $\cR_1$, then at the beginning of its interaction with $\cR_j$ it is in the state $\cL_{j-1}\circ\cdots\circ\cL_1(\rho)$. Hence the corresponding expectation value of the flux observable  is
\[
\langle \Phi_j\rangle = \tr \left( \cL_{j-1}\circ\cdots\circ\cL_1(\rho) \, \Phi_j\right).
\]
It is thus natural to introduce
\begin{equation}\label{def:cyclicflux}
\Phi_j^{cy} := \Lj_1^* \circ \Lj_2^* \circ \cdots \circ \Lj_{j-1}^* (\Phi_j),
\end{equation}
so that $\ds \langle \Phi_j\rangle = \tr \left( \rho\, \Phi_j^{cy}\right).$ We shall call $\Phi_j^{cy}$ the cyclic flux observable associated to $\cR_j$. It is easy to see that one also has
\begin{equation*}\label{eq:cyclicflux}
\Phi_j^{cy} = -\frac{1}{T} \tr_{\rho_{cy}} (U_{cy}^* H_{\En_j} U_{cy} - H_{\En_j}),
\end{equation*}
which is the mean energy variation observable in the reservoir $\cR_j$ during the entire cycle. 

\begin{rema} If we compare the definitions of $\Phi_j$ and $\Phi_j^{cy}$ we can see that the latter contains a supplementary information which is the order of the interactions in a cyclic RIS. Taking two distinct definitions allows us to obtain analogs of the Green-Kubo formula which are similar in both the cyclic and random situation, see (\ref{gkc}), (\ref{gkr}) and (\ref{gkc2}). 
\end{rema}

\begin{rema} Note that the flux observables depend on the thermodynamic parameters but that $\Phi_j^{cy}$ depend on $\zeta_1$, $\zeta_2$, ..., $\zeta_j$ whereas $\Phi_j$ depends only on $\zeta_j$.
\end{rema}

\noindent {\bf Notation.} In the sequel the various results will often have a similar form when written in terms of $\Phi_j^{cy}$ in the cyclic case and of $\Phi_j$ in the random case. When we will write $\Phi_j^\s$, $\s=cy$ or $ra$, the notation $\Phi_j^{ra}$ will therefore stand for $\Phi_j$.

\medskip

\noindent Time reversal plays an important role in linear response theory and the derivation of Green-Kubo formula and Onsager relations. 
Since we consider here averaged flux observables averaged we need to consider also what we call the reversed-time flux observables.
\begin{defi}\label{def:fluxrev} The reversed-time flux observable associated to $\cR_j$ is
\begin{equation*}
\Phi_{j,\rev} := \frac{1}{T} \tr_{\rho_{\En_j}} (U_j H_{\En_j} U_j^* - H_{\En_j}), \quad U_j=\e^{-i\tau_j H_j}.
\end{equation*}
Accordingly the cyclic reversed-time flux observable associated to $\cR_j$ is
\begin{eqnarray}\label{def:cyclicfluxrev}
\Phi_{j,\rev}^{cy} & := & \Lj_{M,\rev}^*\circ\cdots\circ\Lj_{j+1,\rev}^*(\Phi_{j,\rev}) \ = \ \frac{1}{T} \tr_{\rho_{cy}} (U_{cy} H_{\En_k} U_{cy}^* - H_{\En_k}),
\end{eqnarray}
where $\cL_{j,\rev}$ is the reduced dynamics map associated to a time-reversed interaction, i.e. 
\[
\cL_{j,\rev}(\rho)= \tr_{\cH_{\cE_j}} \left( \e^{i\tau_j H_j} \rho\otimes \rho_{\cE_j} \, \e^{- i\tau_j H_j} \right).
\]
\end{defi}

\begin{rema}\label{rem:revRDM} If (\nameref{ergodicity}) and (\nameref{Non-Entanglement}) hold then for any $j=1,\ldots,M$ we actually have 
\begin{equation}\label{eq:revinteraction}
\cL_{j,\rev}^*=\cL_j^{*\rho_{+,\zeta_j}},
\end{equation} 
the $\rho_{+,\zeta_j}$-adjoint of $\cL_j^*$, where we recall that $\rho_{+,\zeta}$ is the global invariant state as given in Assumption (\nameref{Non-Entanglement}). The proof is very similar to the one of (\ref{eq:revLj}) and is left to the reader. This is of course related to the fact that in the Heisenberg picture the $\rho$ adjoint of the dynamics $\tau_t(\cdot) =\e^{itH} \cdot \e^{-itH}$ is $\tau_{-t}$, see Remark \ref{rem:origintrimarkov}. Note also that if Assumption (\nameref{Time-Reversal}) holds it follows directly from the definition of $\cL_{j,\rev}$ that 
\begin{equation}\label{eq:rdmtri}
\cL_{j,\rev} := \Theta \circ \cL_j \circ \Theta,
\end{equation}
i.e. $\cL_{j,\rev}$ is the time-reversal of $\cL_j$ which is in agreement with (\ref{eq:revLj}).
\end{rema}

\begin{rema} 1) The sign discrepancy in the definitions of $\Phi_j$ and $\Phi_{j,\rev}$ takes into account the fact that we have reversed the time. Assume for simplicity that all the interaction times $\tau_j$ are equal to $\tau$. In the limit $\tau\to 0$ of instantaneous interactions it is easy to see that the flux observables $\Phi_j$ and $\Phi_{j,\rev}$ both coincide with 
\[
\Phi_j^{\rm inst}= \frac{1}{M}\tr_{\rho_{\cE_j}} \left( -i[H_j,H_{\cE_j}]\right).
\]

\noindent 2) When (\nameref{Time-Reversal}) holds the fluxes $\Phi_j$ and $\Phi_{j,\rev}$ satisfy the relation 
\begin{equation}\label{eq:phireversal}
\Phi_{j,\rev} = -\Theta \left( \Phi_j\right).
\end{equation} 
In the limit of instantaneous interactions one retrieves the usual relation  $\ds \Phi_j^{\rm inst} = -\Theta \left( \Phi_j^{\rm inst}\right)$, i.e. flux observables are odd with respect to time reversal.

\noindent 3) In the cyclic case, one can notice that $\Phi_{j,\rev}^{cy}$ corresponds to a total time-reversal, where the order of the interactions is reversed as well (compare (\ref{def:cyclicflux}) and (\ref{def:cyclicfluxrev})).
\end{rema}

When Assumption (\nameref{Non-Entanglement}) holds then $H_\cS'+H_{\cE_j}$ is a conserved quantity of the interacting system $\cS+\cE_j$  with $H_\cS'$ is the effective hamiltonian, see Section \ref{ssec:effectiveham}. Hence $\Phi_j$ is naturally the flux observable associated to $H_\cS'$. Namely we have
\begin{prop} Suppose (\nameref{ergodicity}) and (\nameref{Non-Entanglement}) hold. Then 
\begin{equation}\label{translation2}
\Phi_j = \frac{1}{T} (\Lj_j^* (H_{\Sy}') - H_{\Sy}'),
\end{equation}
i.e. $\Phi_j$ is the flux corresponding to the effective energy $H_{\Sy}'$. Similarly one has
\begin{equation}\label{translationrev}
\Phi_{j,\rev} = -\frac{1}{T} (\Lj_{j,\rev}^* (H_{\Sy}') - H_{\Sy}'),
\end{equation}
\end{prop}

\begin{proof} Using (\ref{conservation}) we have $U_j^* H_{\En_j} U_j - H_{\En_j} = H_\cS'-U_j^*H_\cS'U_j$ hence
\[ 
\Phi_j  =\frac{1}{T}\tr_{\rho_{\En_j}}(U_j^*H_{\Sy}'U_j-H_{\Sy}') =\frac{1}{T}(\Lj_j^*(H_{\Sy}')-H_{\Sy}'),
\]
where we have used that $H_\cS'\in\cB(\cH_\cS)$ and the definition of $\cL_j$. This proves (\ref{translation2}).

Then (\ref{translationrev}) follows from (\ref{translation2}) using (\ref{eq:revinteraction}) , (\ref{eq:phireversal}) and Remark \ref{rem:effectivehamreversal}.
\end{proof}


\subsection{Energy conservation and entropy production}

The first central concepts in non-equilibrium systems are energy conservation (1st law) and the relation between entropy production and steady fluxes (2nd law).

\begin{prop}\label{prop:energyconservation}[1st law] If (\nameref{ergodicity}) and (\nameref{Non-Entanglement}) hold  we have
\begin{equation}\label{energy conservation}
\sum_{j = 1}^M \rho_+^\s (\Phi_j^\s) = 0, \quad \s=cy \mbox{ or } ra. 
\end{equation}
\end{prop}

\begin{rema}\label{rem:steadyflux} In the cyclic framework, $\rho_+^{cy}$ is the asymptotic state at the beginning of a cycle and not at the beginning of the interaction between $\cS$ and the $j$-th reservoir (which is thus $\cL_{j-1}\circ\cdots\circ\cL_1(\rho_+^{cy})$). So the steady expectation value of the energy flux out of $\cR_j$ is indeed $\ds \rho_+^{cy} \left( \Lj_1^* \circ \Lj_2^* \circ \cdots \circ \Lj_{j-1}^* (\Phi_j) \right) = \rho_+^{cy}(\Phi_j^{cy})$.
\end{rema}

\begin{proof} Using (\ref{def:cyclicflux}),  (\ref{translation2}) and $\cL_{cy}(\rho_+^{cy})=\rho_+^{cy}$ we have
\[
\sum_{j=1}^M \rho_+^{cy}(\Phi_j^{cy}) \ = \ T\sum_{j=1}^M \rho_+^{cy}\left( \Lj_1^*\circ\cdots\circ\Lj_{j-1}^*(\Lj_j^*(H_{\Sy}')-H_{\Sy}')\right)
 \ = \ T \rho_+^{cy}\left(\Lj_{cy}^*(H_{\Sy}')-H_{\Sy}'\right) 
 \ = 0.
\]
Similarly, recall $\ds\cL_{ra}=\frac{1}{M} (\cL_1+\cdots+\cL_M)$, so that using (\ref{translation2}) and $\cL_{ra}(\rho_+^{ra})=\rho_+^{ra}$ we get
\[
\sum_{j=1}^M\rho_+^{ra}(\Phi_j)  = T \sum_{j=1}^M \rho_+^{ra}\left(\Lj_j^*(H_{\Sy}')-H_{\Sy}'\right) = TM \rho_+^{ra} \left( \cL_{ra}^*(H_\cS') -H_{\Sy}'\right) =0.
\]
\end{proof}

\begin{rema}\label{rem:zerofluxeq} Of course, as expected and mentioned at the beginning of the proof of Proposition \ref{prop:NEandequilibrium}, when all the temperatures are equal each steady flux vanishes, i.e. $\rho_+^\s(\Phi_j^\s)=0$ for all $j$. 
\end{rema}

We now turn to entropy production. Recall that the entropy production during an interaction between $\cS$ and a system $\cE$ is given by (\ref{def:entropyprod}). For a non-equilibrium RIS, and if  $\bj:=(j_n)_{n\in\N^*}$ describes the sequence of interactions (see Section \ref{sec:noneqris}), the entropy production during the $n^{th}$ interaction is thus
\begin{equation*}
 \sigma_n(\bj):=\Ent(U_{j_n}\times\rho_{n-1}(\bj)\otimes\rho_{\En_{j_n}}\times U_{j_n}^* \, | \, \rho_n(\bj)\otimes\rho_{\En_{j_n}}),
\end{equation*}
where $\rho_n(\bj)=\Lj_{j_n}\circ\cdots\circ\Lj_{j_1}(\rho)$ denotes the state of $\cS$ after $n$ interactions. Using (\ref{ebe}) we have the following entropy balance equation for the $N$ first interactions
\begin{equation}\label{ebe2}
 \sum_{n=1}^N\sigma_n(\bj)=\Ent(\rho_N(\bj))-\Ent(\rho)-T\sum_{n=1}^N\beta_{j_n} \tr\left(\Lj_{j_{n-1}}\circ\cdots\circ\Lj_{j_1} (\rho)\, \times \Phi_{j_n}\right).
\end{equation}
Consequently, in the random, resp. cyclic, cases we can define the entropy production 
\begin{equation}\label{def:entropyprod2}
\sigma_{ra}(\bj,N):= \sum_{n=1}^N\sigma_n(\bj), \quad \mbox{ resp. } \sigma_{cy}(N):=\sum_{n=1}^{NM}\sigma_n(\bj),
\end{equation}
associated to the $N$ first interactions, resp. cycles. For cyclic interactions (\ref{ebe2}) becomes
\begin{eqnarray}\label{ebecy}
 \sigma_{cy}(N) & = & \Ent\left(\Lj_{cy}^N(\rho)\right)-\Ent(\rho)-T\sum_{j=1}^M\sum_{n=1}^N\beta_j\tr\left( \cL_{cy}^{n-1}(\rho) \times \Phi_j^{cy} \right).
\end{eqnarray}

\begin{defi} We define the asymptotic entropy production rate of a cyclic or random Repeated Interaction System by
 \begin{equation*}
  \sigma_{cy}^+:=\lim_{N\rightarrow+\infty}\frac{\sigma_{cy}(N)}{NT}, \quad  \sigma_{ra}^+(\bj):=\lim_{N\rightarrow+\infty}\frac{\sigma_{ra}(\bj,N)}{N\frac{T}{M}},
 \end{equation*}
provided the limits exist.
\end{defi}

\begin{prop}\label{prop:2ndlaw}[2nd law] Assume (\nameref{ergodicity}) holds. Then the asymptotic entropy productions exist and the following second laws of thermodynamics hold
\begin{equation*}
 \sigma_{cy}^+=-\sum_{j=1}^M\beta_j\rho_+^{cy} (\Phi_j^{cy}) \quad \mbox{ and }\quad   \sigma_{ra}^+(\bj)=-\sum_{j=1}^M\beta_j\rho_+^{ra}(\Phi_j) \quad \cP -a.s.
 \end{equation*}
\end{prop}

\noindent {\it Proof.} Since $\cH_\cS$ has finite dimension the von Neumann entropies $\Ent\left(\Lj_{cy}^N(\rho)\right)$ and $\Ent\left(\rho_N(\bj)\right)$ are uniformly bounded hence give no contribution to the asymptotic entropy production. In the cyclic case the result thus follows directly from (\ref{ebecy}) and Proposition \ref{erprop}.

In the random case, using (\ref{ebe2}), Proposition \ref{erprop} and Theorem \ref{rd} with $A(j):=\beta_j\Phi_j$ we get that $\sigma_{ra}^+(\bj)$ indeed exists $\cP$-almost surely and is given by
\[
\sigma_{ra}^+(\bj) = - M \rho_+^{ra} \left( \E_p(\beta\Phi)\right) = - \sum_{j=1}^M \beta_j  \rho_+^{ra} \left(\Phi_j\right). 
\]

\begin{rema}\label{def:timescale} In the cyclic case we have divided by $T$ which is the total duration of a cycle while in the random case we have divided by $\frac{T}{M}$ which is the mean duration of an interaction, in agreement with (\ref{def:entropyprod2}). These different scalings will appear regularly in the sequel. 

In the random case we could have equivalently divided by the total duration $\sum_{n=1}^N \tau_{j_n}$ instead of $N\frac{T}{M}$ because $\frac{1}{N}\sum_{n=1}^N \tau_{j_n}\to \frac{T}{M}$ by the strong law of large numbers.
\end{rema}


\subsection{Green-Kubo and Onsager}\label{lrt}

The next step is linear response theory which is concerned with the response of the system to a small perturbation from equilibrium. In this section we derive the analog of the Green-Kubo  fluctuation-dissipation formula which relates the transport coefficients of the system out of equilibrium to flux-flux correlation at equilibrium. Linear response theory will be completed in Section \ref{sec:fluctuationthm} by the Central Limit, \emph{aka} Fluctuation-Dissipation, Theorem \ref{clt}.

In order to state the Green-Kubo formula we first recall the notion of dissipation function. Recall that the characteristic time of the system is $T$.
\begin{defi} The dissipation function associated to $\cR_j$ is 
\begin{equation}\label{def:dissipation}
\cD_j(X,Y) := \frac{1}{T}\left(\cL_j^*(X^*Y)-\cL_j^*(X^*)Y-X^*\cL_j^*(Y)+X^*Y\right).
\end{equation}
\end{defi}
\noindent If the Kraus decomposition of $\cL_j$ is given by (\ref{eq:kraus}), it is easy to see that  $\ds \cD_j(X,Y) = \frac{1}{T}\sum_{i\in I} \, [V_i,X]^*[V_i,Y]$. In particular $\ds \cD_j(X,X) = \frac{1}{T}\sum_{i\in I} \, [V_i,X]^*[V_i,X]$ is non-negative and is zero only if $X$ commutes with all the $V_i$'s.

\begin{rema} Dissipation function has been introduced in \cite{Li} for continuous-time quantum dynamical semigroups $\left(\e^{tL}\right)_{t\geq 0}$ as the sesquilinear map $D$ acting on $\cB(\cH_\cS)$ and defined by
\[
D(X,Y) = L(X^*Y)-L(X^*)Y-X^*L(Y).
\]
Here $L$ is the Lindblad generator of a semigroup of unital completely positive maps, i.e. corresponding to the Heisenberg picture. If we fix some characteristic time $T$, according to (\ref{def:dissipation}) the dissipation function associated to the unital completely positive map $\e^{T L}$ is
\[
\cD_T(X,Y) := \frac{1}{T}\left(\e^{T L}(X^*Y)-\e^{T L}(X^*)Y-X^*\e^{T L}(Y)+X^*Y\right)
\]
and it is easy to see that $\ds \lim_{T\to 0} \cD_T(X,Y)=D(X,Y)$.
\end{rema}

\begin{theo}[GK formula and Onsager relations]\label{gko} Suppose (\nameref{ergodicity}) and (\nameref{Non-Entanglement}) hold.  

\noindent 1) The maps $\bfzeta \mapsto \rho^\s_{+,\bfzeta}$, $\bfzeta \mapsto \Phi_{j,\bfzeta}$ and $\bfzeta \mapsto \Phi^{cy}_{j,\bfzeta}$ are infinitely differentiable. The quantities 
\begin{equation*}\label{def:kinetic coef}
L_{jk}^\s := \partial_{\zeta_k} \rho^\s_{+,\bfzeta} (\Phi^\s_{j,\bfzeta}) \big\lceil_{\bfzeta = 0},   \quad \s = cy \mbox{ or } ra, 
\end{equation*} 
are called the kinetic coefficients.

\noindent 2) We have the following analogs of the Green-Kubo formula
\begin{eqnarray}
L_{jk}^{cy} & = & T \sum_{n=0}^{+\infty} \rho_+\Big( \cL_{k+1}^*\circ\cdots\circ \cL_M^*\circ \Lj_{cy}^{*n} \circ \cL_1^*\circ\cdots\circ\cL_{j-1}^* (\Phi_j) \times \Phi_{k,\rev} \Big) \label{gkc} \\ & & \ + \delta_{j>k}\, T\rho_+\Big( \Lj_{k+1}^*\circ\cdots\circ\Lj_{j-1}^*(\Phi_j) \times \Phi_{k,\rev}\Big) +\frac{1}{2}\delta_{jk} \rho_+(\cD_j(H_\Sy ',H_\cS')),\nonumber \\
L_{jk}^{ra} & = & \frac{T}{M} \sum_{n=0}^{+\infty}\rho_+(\Lj_{ra}^{*n}(\Phi_j)\Phi_{k,\rev}) + \frac{1}{2}\delta_{jk} \rho_+(\cD_j(H_\Sy ',H_\cS')),\label{gkr}
\end{eqnarray}
where all the quantities on the right-hand side are calculated at equilibrium $\bfzeta=0$, e.g. $\rho_+$ stands for $\rho_+=\rho_{+,0}^{cy}=\rho_{+,0}^{ra}$ (see Remark \ref{rem:NEequilibriumstate}).

\noindent 3) If moreover Assumption (\nameref{Time-Reversal}) holds, then the following analogs of the Onsager reciprocity relations are satisfied:
\begin{equation}\label{ons}
L_{jk}^{cy} = L_{kj}^{rcy} \ \mbox{ and } \ L_{jk}^{ra} = L_{kj}^{ra},
\end{equation}
where $L_{kj}^{rcy}$ denotes the kinetic coefficient associated to the cyclic RIS in which we reverse the order of the interactions, i.e. $\cL^{rcy}=\cL_1\circ\cdots\circ \cL_M$. In the specific case $M=2$, we retrieve the usual Onsager reciprocity relations $L_{jk}^{cy}=L_{kj}^{cy}$ for the cyclic case too.
\end{theo}

\noindent Let us comment on identites (\ref{gkc}) and (\ref{gkr}). 

i) The sum in (\ref{gkc}) describes the flux-flux correlation at equilibrium between reservoirs $j$ and $k$. Since the $\cR_j$ interacts only once during each cycle, using (\ref{def:cyclicflux})-(\ref{def:cyclicfluxrev}), it is easy to see that this sum can actually be written in terms of  $\cL_{cy}$ leading to the more condensed form
\begin{equation}\label{gkc2}
\sum_{n=0}^{+\infty} \rho_+ \Big( \cL_{k+1}^*\circ\cdots\circ \cL_M^*\circ \Lj_{cy}^{*n} \circ \cL_1^*\circ\cdots\circ\cL_{j-1}^* (\Phi_j) \times \Phi_{k,\rev} \Big) = \sum_{n=0}^{+\infty} \rho_+ (\Lj_{cy}^{*n} (\Phi_j^{cy}) \Phi_{k,\rev}^{cy}).
\end{equation}

ii) The sum (\ref{gkc2}) takes into account only correlations when at least one cycle has been achieved. The second term in the right-hand side of (\ref{gkc}) takes into account those contribution of the flux-flux correlation between reservoirs $k$ and $j$ if less than a cycle has already occured, which can happen only if $j>k$. 

iii) In both (\ref{gkc}) and (\ref{gkr}) the $\delta_{jk}$ term takes into account the self-correlation of reservoir $j$ with itself during its first interaction with $\cS$. It is non-negative and vanishes only if $\cS$ is not effectively coupled to $\cE_j$ in the following sense. If $\cL_j=\sum_i V_i \cdot V_i^*$ then $\rho_+(\cD_j(H_\Sy ',H_\cS'))=0$ iff $\cD_j(H_\Sy ',H_\cS')=0$ (the latter is non-negative and $\rho_+$ is positive definite) which in turn holds iff $[V_i,H_\cS']=0$ for all $i$. But this implies that $\cL_j^*(H_\cS')=\sum_i V_i^*H_\cS'V_i = \sum_i V_i^*V_iH_\cS'=H_\cS'$ because $\cL_j^*$ is unital, and hence $\Phi_j=0$ by (\ref{translation2}): whatever are the various temperatures and the state of $\cS$ there is no flux between $\cS$ and $\cR_j$. We mention that a similar term appears in \cite{JPW} in the context of quantum dynamical semigroups. 

iv) Finally the prefactors $T$ and $\frac{T}{M}$ represent the ``time-step''. In the cyclic case, $\cL_{cy}$ describes the evolution during a cycle while $\cL_{ra}$ is associated only to a time-step $\frac{T}{M}$.

\begin{rema} The Onsager reciprocity relations are satisfied only for the random and the $M=2$ cyclic case. In the $M>2$ cyclic case Assumption  (\nameref{Time-Reversal}) reverses the time only locally, leaving the order of the interaction unchanged (see Remark \ref{pwtr}). This explains why $L_{jk}^{cy}$ has to be compared to the kinetic coefficient $L_{kj}^{rcy}$ of the reversed-order cyclic model.
\end{rema}

\noindent {\bf Notation} In the rest of the paper, any quantity with a subscript/superscript $rcy$ should be understood as the quantity associated to the reversed-order cyclic RIS in which $\cS$ interacts first with $\cR_M$, then $\cR_{M-1}$,... For example $\cL_{rcy}=\cL_1\circ\cdots\circ \cL_M$, $\Phi_j^{rcy}= \cL_M^*\circ\cdots\circ \cL_{j+1}^*(\Phi_j)$. Of course, any result which holds true for the cyclic RIS also holds for the reversed cyclic one.


\section{Entropic fluctuation}\label{fluc}

The purpose of this section is to go beyond Proposition \ref{prop:2ndlaw} and to study the statistical fluctuations of the entropy fluxes $\beta_j\Phi_j$ going out the reservoirs $(\Res_j)_{1\leq j\leq M}$.


\subsection{Full counting statistics}

To study the entropy fluctuation, following e.g. \cite{JOPP,BDBP,benoist2015full,BPP,BPR}, we consider the statistics of the increments of the entropy observable as given by a two time measurement protocol, also called Full Couting Statistics (FCS), and which we briefly recall for the convenience of the reader. Consider a quantum system $\Sy$ with underlying finite dimensional Hilbert space $\Hi$ and let $A$ be the observable of interest. Suppose the system is in the state $\rho$ when we perform a first measurement of $A$. The possible outcomes of the measurement are eigenvalues of $A$ and $a\in \sp(A)$ is observed with probability 
\begin{equation*}
\mb{P} (A = a) = \tr(\Pi_a (A) \rho \Pi_a (A))
\end{equation*}
where $\Pi_a (A)$ is the spectral projector of $A$ associated to the eigenvalue $a$. After this first measurement the state of the system reduces to $\ds \frac{\Pi_a (A) \rho \Pi_a (A))}{\tr(\Pi_a (A) \rho \Pi_a (A))}$. Subsequently, if the evolution of the system during some time interval of length $\tau$ is described by some unitary operator $U$, a second measurement of $A$ at time $\tau$ gives the value $a'\in \sp(A)$ with probability
$$
\frac{\tr \big( \Pi_{a'}(A) U \Pi_a(A) \rho \Pi_a(A) U^*\Pi_{a'}(A)\big)}{\tr(\Pi_a(A) \rho)}.
$$
The joint probability distribution of the two measurements is thus given by 
$$
\P(a,a') = \tr \big( \Pi_{a'}(A) U \Pi_a(A) \rho \Pi_a(A) U^*\Pi_{a'}(A)\big),
$$ 
and the statistics of the increment $\Delta A$ of $A$ as given by this protocol is therefore given by
\begin{equation*}
\mb{P} (\Delta A=\delta) = \sum_{\substack{a, a' \in \sp(A) \times \sp(A) \\ a' - a = \delta}}  \tr \big( \Pi_{a'}(A) U \Pi_a(A) \rho \Pi_a(A) U^*\Pi_{a'}(A)\big).
\end{equation*}
Note that if $A$ commutes with the initial state $\rho$ we simply have 
\begin{equation}\label{eq:probafcs}
\mb{P} (\Delta A=\delta) = \sum_{a,a' \, |\, a' - a = \delta}  \tr \big( \Pi_{a'}(A) U \Pi_a(A) \rho U^* \big)
\end{equation} 
so that the expectation of $\Delta A$ as given by this two-measurement protocol is 
\begin{eqnarray}\label{eq:fcsexpectation}
\mb{E}_{\mb{P}} (\Delta A) & = &\!\!\! \sum_{a, a' \in \sp(A) \times \sp(A)} \!\!\! (a'-a)\, \tr \left( \Pi_{a'}(A) U \Pi_a(A) \rho \, U^*\right) 
 \ = \ \tr \big( \rho (U^*AU-A)\big),
\end{eqnarray}
and coincides with the expectation value of the flux observable $U^*AU-A$ associated to $A$.

Note also that, although it does not appear in the notation, the probability law $\P$ depends on the initial state $\rho$ of the system when the first measurement is performed.

\medskip

In this section, we are interested in the full statistics of the increment of entropy observables of the probes, where the entropy observables are defined in Section \ref{entropy}. Since probes are initially in thermal equilibrium, up to an irrelevant constant the entropy observable of  $\cE_j$ is  $S_{\cE_j} := \beta_j H_{\cE_j}$. The corresponding entropy increment observable is therefore $U_j^*S_{\cE_j}U_j-S_{\cE_j}= \beta_j (U_j^*H_{\cE_j}U_j-H_{\cE_j})$. We shall also consider the energy increment observables $U_j^*H_{\cE_j}U_j-H_{\cE_j}$ which up to a prefactor $\frac{1}{T}$ correspond to the flux observables considered in Section \ref{sec:fluxobs}.

Let now $\boldsymbol{j} :=(j_n)_{n\in\N^*}\in\{1,\ldots, M\}^{\N^*}$ be the sequence of indices describing the sequence of probes with which $\cS$ interacts, and denote by $\En_{j_n}^n$ the $n^{th}$ probe. Recall that $\En_{j_n}^n$ is a copy of $\En_{j_n}$. Then, according to (\ref{eq:probafcs}), the probability distribution of the increment of entropy for the $n$-th interaction is given by
\begin{equation*}
\mb{P} \left(\Delta \mc{S}_{\En_{j_n}^n} = \varsigma\right) = \sum_{s' - s = \varsigma}  \tr\left(\Pi_{s'}\left(S_{\En_{j_n}^n}\right)U_n\Pi_s \left(S_{\En_{j_n}^n}\right) \times \rho \otimes \rho_{\En_{j_n}^n} \times U_n^* \right), 
\end{equation*}
where the sum runs over $s,s' \in \sp \left(S_{\En_{j_n}^n}\right)$ such that $s' - s = \varsigma$, $\rho$ is the state of $\cS$ at the beginning of the interaction and we have used that the observable  $S_{\En_{j_n}^n}$ commutes with the ``initial state'' $\rho \otimes \rho_{\En_{j_n}^n}$.  

To describe the statistics of entropy increments during the $N$ first interactions let us introduce first some notation. For $\bfs := ((s_n,s'_n))_n \in (\mb{R}^2)^{\mb{N}^*}$ and, for all $N\geq 1$, we define
\begin{equation*}
 \mc{U}_N^\bfs (\bj):= U_{(s_N,s'_N)}^N \left(S_{\En_{j_N}}^N\right) \times \cdots \times U_{(s_1,s'_1)}^1 \left(S_{\En_{j_1}}^1\right), \quad \mbox{where } \ U_{(s,s')}^n(S)=\Pi_{s'}(S)U_n\Pi_s(S)
\end{equation*}
Note that clearly $U_{(s,s')}^n(S)$ is non-zero only when $s,s'\in \sp(S)$. 
Then the statistics of entropy increments during the $N$ first interactions is given by
\begin{equation*}
\mb{P}\left(\Delta \mc{S}_{\En_{j_1}^1} = \varsigma_1, \cdots, \Delta \mc{S}_{\En_{j_N}^N} = \varsigma_N\right)=\sum\tr\left(\U_N^s(\bs{j})\times\rho_{tot}^N(\bs{j})\times\U_N^{s}(\bs{j})^*\right),
\end{equation*}
where the sum runs over $\bfs \in (\mb{R}^2)^{N}$ such that $\forall 1 \leq n \leq N, s_n, s'_n \in \sigma\left(S_{\En_{j_n}^n}\right) \times \sigma\left(S_{\En_{j_n}^n}\right)$ and $s'_n - s_n = \varsigma_n$, and where $\rho_{tot}^N(\bs{j})$ denotes the total initial state $\rho\otimes\rho_{\En_{j_1}^1}\otimes\cdots\otimes\rho_{\En_{j_N}^N}$.

\begin{defi}\label{def:randomincrements} The (random) vectors of entropy increments and energy increments after $N$ interactions are 
$\mc{S}_{\Res}^{N}(\bs{j}):=\left(\mc{S}_{\Res_1}^{N}(\bs{j}),\cdots,\mc{S}_{\Res_M}^{N}(\bs{j})\right)$, $\Q_{\Res}^{N}(\bs{j}):=\left(\Q_{\Res_1}^{N}(\bs{j}),\cdots,\Q_{\Res_M}^{N}(\bs{j})\right)$ where
\begin{equation*}
\mc{S}_{\Res_j}^{N}(\bs{j}):=\sum_{\substack{1\leq n\leq N\\j_n=j}}\Delta\mc{S}_{\En_{j_n}^n}, \quad \Q_{\Res_j}^{N}(\bs{j}):= - \beta_j^{-1}\mc{S}_{\Res_j}^{N}(\bs{j}).
\end{equation*}
\end{defi}

\begin{rema} The sign convention in the definition of $\Q_{\Res_j}^{N}(\bs{j})$ is in agreement with the one of the flux observables $\Phi_j$ of Section \ref{sec:fluxobs}.
\end{rema}

We obviously have the following 
\begin{prop} The joint probability distribution of $\mc{S}_{\Res}^{N}(\bs{j})$ is given by
\begin{equation}\label{eq:entropyincrementproba}
\mb{P} \left(\mc{S}_{\Res}^{N}(\bs{j})=\bs\varsigma\right)=\sum\tr\left(\U_N^s(\bs{j}) \times \rho_{tot}^N(\bs{j}) \times \U_N^{s}(\bs{j})^*\right),
\end{equation}
where the sum runs over all $\bfs$ such that for all $1 \leq n \leq N, s_n, s'_n \in \sigma\left(S_{\En_{j_n}^n}\right) \times \sigma\left(S_{\En_{j_n}^n}\right)$ and for any $j\in\{1,\ldots,M\}$ one has
$\ds \sum_{\substack{1 \leq n \leq N\\j_n = j}} s'_n - s_n = \varsigma_j$ with $\bs\varsigma=(\varsigma_1,\ldots,\varsigma_M)$.
\end{prop}

In the case of cyclic RIS, in agreement with (\ref{ebecy}), we shall actually consider
\begin{equation*}\label{def:cyclicincrements}
\mc{S}_{\Res}^{N}(cy):=\mc{S}_{\Res}^{N\times M}(\bs{j}^{cy}), \mbox{ resp. } \mc{Q}_{\Res}^{N}(cy):=\mc{Q}_{\Res}^{N\times M}(\bs{j}^{cy}),
\end{equation*}
the vector of entropy, resp. energy, increments associated to the cyclic RIS after $N$ cycles. 

In the random case note that randomness is now twofold: randomness due to the random order of interactions and a ``quantum'' randomness due to the double measurement protocol, and that the latter depends on the former. Indeed, as already mentioned the probability law $\P$ of a double measurement protocol depends on the initial state of the (entire) system, hence on $\rho$ but more importantly on the sequence $\bj$ of probes (via the $\rho_{\cE^n_{j_n}}$'s). To be precise, given a sequence $\bj$ we should denote $\P_\bj$ instead of $\P$ and then the joint probability distribution with respect  to the two alea is
\begin{equation}
\hat\P(\bj,\bfs) = \cP(\bs{j})\times \P_\bj(\bfs).
\end{equation}


\subsection{Moment generating function}

We will analyze the large time limit, i.e. $N\to\infty$, statistics of the entropy and energy increments through their moment generating functions (MGF). 
\begin{defi}[Moment generating function] Let $\al \in \C^M$. We denote by $r_{N,\rho}^{\bj} (\al)$ and $r_{N,\rho}^{cy} (\al)$ the respective MGF of the vectors $-\mc{S}_{\Res}^{N}(\bs{j})$ and $-\mc{S}_{\Res}^{N}(cy)$ at $\al$, i.e. 
\begin{equation*}
r_{N,\rho}^{\bj} (\al) := \mb{E}_{\mb{P}_\bj} \left(e^{-\sum_{j=1}^M \al_j \mc{S}_{\Res_j}^{N}(\bj)}\right) \mbox{ and } 
r_{N,\rho}^{cy} (\al) := \mb{E}_{\mb{P}} \left(e^{-\sum_{j=1}^M \al_j \mc{S}_{\Res_j}^N(cy)}\right).
\end{equation*}
\end{defi}

Note that the above MGF are defined with respect to the ``quantum'' alea. We have also stressed their dependence on the initial state $\rho$ of the small system. In the random RIS case we shall also consider the MGF with respect to the two alea. 
\begin{defi} For $\al \in \C^M$, let
$\ds 
r_{N,\rho}^{ra} (\al) := \E_{\cP} \left(r_{N,\rho}^{\bj} (\al)\right) = \mb{E}_{\hat\P} \left(e^{-\sum_{j=1}^M \al_j \mc{S}_{\Res_j}^{N}(\bj)}\right).
$
\end{defi}

To understand the MGF $r_{N,\rho}^{\bj} (\al)$, $r_{N,\rho}^{cy} (\al)$ and $r_{N,\rho}^{ra} (\al)$, following \cite{JPW,BDBP,BPP}, we consider the following deformations of the reduced dynamics maps $(\Lj_j)_{1 \leq j \leq M}$.
\begin{defi}\label{defi:Lalpha} For $\al=(\alpha_1,\ldots,\alpha_M) \in \C^M$ we define
\begin{equation*}\label{def:deformedrdm}
 \Lj_j^{[\al]*}(X) := \tr_{\En_j} \big(\un \otimes \rho_{\En_j}^{1-\al_j} \times U_j^* \times X \otimes \rho_{\En_j}^{\al_j} \times U_j\big),
\end{equation*}
and
\begin{equation}\label{def:deformedrdmcyra}
\Lj_{cy}^{[\al]*} := \Lj_1^{[\al]*} \circ \cdots \circ \Lj_M^{[\al]*},\quad \Lj_{ra}^{[\al]*} := \frac{1}{M} (\Lj_1^{[\al]*} + \cdots + \Lj_M^{[\al]*}).
\end{equation}
\end{defi}

\begin{rema} The notation $ \Lj_j^{[\al]*}$ is chosen in order to be consistent with the reduced dynamics maps $\cL_j^*$. Indeed, $\Lj_j^{[0]*}=\cL_j^*$ and $\Lj_{\s}^{[0]*}=\cL_\s^*$, $\s=cy$ or $ra$. Of course  $\Lj_j^{[\al]}$ will denote its dual map so that $\Lj_j^{[0]}=\cL_j$. Note also that for any $j$ the map $\cL_j^{[\al]*}$ only depends on $\alpha_j$.  
\end{rema}

\noindent The connection between the MGF and the $\cL_j^{[\al]*}$'s is provided by the following proposition. For the convenience of the reader we briefly sketch its proof in Section \ref{sec:proofmgfdeformedrdm}.

\begin{prop}\label{prop:mgfdeformedrdm} For any $\al\in \C^M$ and $N\geq 1$ one has
\begin{equation}\label{specmgf}
r_{N,\rho}^{\bj} (\al) = \rho\left( \Lj_{j_1}^{[\al]*} \circ \cdots \circ \Lj_{j_N}^{[\al]*}(\un) \right).  
\end{equation}
As a consequence
\begin{equation}\label{specmgfcyra}
r_{N,\rho}^{cy} (\al) =  \rho\left( \left(\Lj_{cy}^{[\al]*}\right)^N (\un) \right)   \ \mbox{ and } \  r_{N,\rho}^{ra} (\al) = \rho\left( \left(\Lj_{ra}^{[\al]^*}\right)^N (\un) \right).    
\end{equation}
\end{prop}

\noindent Of course the same approach can be used to study the energy increments.
\begin{defi} Let $\al \in \C^M$. We denote by $\tilde r_{N,\rho}^{\bj} (\al)$ and $\tilde r_{N,\rho}^{cy} (\al)$ the respective MGF of the vectors $\mc{Q}_{\Res}^{N}(\bs{j})$ and $\mc{Q}_{\Res}^{N}(cy)$ at $\al$, i.e. 
\begin{equation*}
\tilde r_{N,\rho}^{\bj} (\al) := \mb{E}_{\mb{P}} \left(e^{\sum_{j=1}^M \al_j \mc{Q}_{\Res_j}^{N}(\bj)}\right) \ \mbox{ and } \
\tilde r_{N,\rho}^{cy} (\al) := \mb{E}_{\mb{P}} \left(e^{\sum_{j=1}^M \al_j \mc{Q}_{\Res_j}^N(cy)}\right),
\end{equation*}
as well as, in the random case, the MGF with respect to the two alea, i.e.
\[
\tilde r_{N,\rho}^{ra} (\al) := \E_{\cP} \left(\tilde r_{N,\rho}^{\bj} (\al)\right) = \mb{E}_{\mb{P},\cP} \left(e^{\sum_{j=1}^M \al_j \mc{Q}_{\Res_j}^{N}(\bj)}\right).
\]
\end{defi}
\noindent It follows directly from Definition \ref{def:randomincrements} that the entropy and energy MGF satisfy the relations
\begin{equation}\label{eq:mgfconnection}
\tilde{r}_{N,\rho}^{\bj} (\al) = r_{N,\rho}^{\bj} \left( \frac{\al}{\beta} \right), \qquad \tilde{r}^{\s}_{N,\rho} (\al) = r^{\s}_{N,\rho} \left( \frac{\al}{\beta} \right), \ \s=cy,\, ra,
\end{equation}
and where $\ds \frac{\al}{\beta}$ denotes the vector $\ds\left(\frac{\al_1}{\beta_1},\ldots,\frac{\al_M}{\beta_M}\right) $.

Although not mentioned explicitly all the above quantities depend on the vector $\bfzeta=(\zeta_1,\ldots,\zeta_M)$ of thermal forces and, as for $\cL_j$, the map $\cL_{j,\bfzeta}^{[\al]}$ actually only depends on $\zeta_j$.

\begin{prop}\label{deformedsp} If (\nameref{ergodicity}) holds then, for any $\bfzeta,\al\in\R^M$, the map $\Lj_{\s,\bf\zeta}^{[\al]*}$ is a primitive CP map. We shall denote by $r_{\bfzeta}^{\s}(\al)>0$ its spectral radius.
\end{prop}

\begin{proof} The proof is the same as the one of Proposition \ref{erprop}.
\end{proof}

\noindent One can then relate the large $N$ behaviour of the moment generating function to the spectral radius of $\Lj_{\s}^{[\al]*}$ using Proposition \ref{prop:mgfdeformedrdm}.
\begin{theo}\label{theojpw} Suppose (\nameref{ergodicity}) holds. 

\noindent 1) For any initial state $\rho$ and $\alpha,\bfzeta\in \R^M$ one has
\begin{equation}\label{eq:largetimemgf}
r^{\s}_\bfzeta(\al) = \lim_{n \rightarrow +\infty} r_{n,\rho}^{\s} (\al)^{\frac{1}{n}}.
\end{equation}

\noindent 2) If (\nameref{Time-Reversal}) holds, then we have the following version of the Evans-Searles symmetry:
\begin{equation}\label{ES}
r^{cy}_\bfzeta({\bf 1} - \al) = r^{rcy}_\bfzeta (\al), \quad r^{ra}_\bfzeta({\bf 1} - \al)=r^{ra}_\bfzeta(\al)
\end{equation}
where $r^{rcy}_\bfzeta (\al)$ is the spectral radius of $\Lj_{rcy,_\bfzeta}^{[\al]*}:=\Lj_{M,_\bfzeta}^{[\al]*}\circ\cdots\circ\Lj_{1,_\bfzeta}^{[\al]*}$, the deformed reduced dynamics map corresponding to the cyclic RIS with a reversed order, and ${\bf 1}=(1,\ldots,1)$.

\noindent 3) If (\nameref{Non-Entanglement}) holds we have the following translation symmetry
\begin{equation}\label{ts}
r^\s_\bfzeta(\alpha) = r^\s_\bfzeta(\alpha+\lambda \beta^{-1}), \quad \s=cy,\, ra,
\end{equation}
for any $\alpha,\bfzeta\in\R^M$, $\lambda\in\R$ and where $\beta^{-1}=(\beta_1^{-1},\ldots,\beta_M^{-1})$. 
\end{theo}

In view of (\ref{eq:largetimemgf}) we will call $\al \mapsto r^{\s}_\bfzeta(\al)$ the \textit{large time moment generating function} of the entropy fluxes. Of course (\ref{eq:mgfconnection}) immediately leads to a similar result for energy fluxes, i.e.
\begin{equation*}
\tilde r^{\s}_\bfzeta(\al) := \lim_{n \rightarrow +\infty} \tilde r_{n,\rho}^{\s} (\al)^{\frac{1}{n}}  = r^{\s}_\bfzeta\left(\frac{\al}{\beta}\right).
\end{equation*}
In particular the translation symmetry (\ref{ts}) becomes 
\begin{equation}\label{eq:gaspard}
\tilde r^\s(\alpha) = \tilde r^\s(\alpha+\lambda {\bf 1})
\end{equation}
for any $\alpha\in\R^M$ and $\lambda\in\R$. This translation symmetry (\ref{eq:gaspard}) is related to the energy conservation (\ref{energy conservation}), see Remark \ref{rem:gaspard} below. It is thus not surprising that it requires  (\nameref{Non-Entanglement}).


\subsection{First moments and link with linear response theory}

Our next result concerns the first and second moment of the entropy/energy increments and makes the connection with the results of Section \ref{lrt}. We thus suppose throughout this section that (\nameref{ergodicity}) and (\nameref{Non-Entanglement}) hold. We also introduce the following maps which act on $\cB(\cH_\cS)$.
\begin{defi} For any $j=1,\ldots,M$ let 
\begin{equation}\label{def:fluxmap}
\varphi_{j}(X) = \frac{1}{T} \left( \cL_{j}^*(XH_\cS') - \cL_{j}^*(X)H_\cS'\right), \quad \varphi_{j}^{cy} = \cL_1^*\circ\cdots\circ \cL_{j-1}^*\circ \varphi_j \circ \cL_{j+1}^*\circ\cdots\circ\cL_M^*.
\end{equation}
\end{defi}

\begin{rema}\label{rem:fluxmap} If follows from (\ref{def:cyclicflux}) and (\ref{translation2}) that  $\varphi_{j}(\un)=\Phi_j$ and $\varphi_j^{cy}(\un)=\Phi_j^{cy}$.
\end{rema}

The maps $\varphi_j$ appear naturally as the derivatives with respect to $\alpha$ of the deformed maps $\cL_{j}^{[\alpha]*}$, see Lemma \ref{lem:Lalphaderiv}. Of course, similarly to the flux observable 
$\Phi_j$, the map $\varphi_j$ depends on the thermodynamical force $\zeta_j$. It then follows directly from (\ref{eq:revinteraction}) that for any $\zeta\in\R$ one has
\[
\varphi_{j,\zeta}^{\rho_{+,\zeta}}(X) = \frac{1}{T} \left(\cL_{j,\rev}^*(X)H_\cS' - \cL_{j,\rev}^*(XH_\cS') \right)
\]
so that in particular, using (\ref{translationrev}), we get
\begin{equation}\label{eq:fluxmap}
\varphi_{j,\zeta}^{\rho_{+,\zeta}}(\un)=\Phi_{j,\rev}.
\end{equation}

\begin{theo}\label{momentstheo} Suppose (\nameref{ergodicity}) and (\nameref{Non-Entanglement}) hold. Then we have

\noindent 1) For all $j=1,\ldots,M$,
\begin{equation}\label{first moment}
\frac{\partial r^{cy}(\alpha)}{\partial {\al_j}}\lceil_{\alpha=0} =\beta_jT \rho_{+}^{cy}(\Phi_{j}^{cy}), \quad \frac{\partial r^{ra}(\alpha)}{\partial {\al_j}}\lceil_{\alpha=0} =\beta_j\frac{T}{M} \rho_{+}^{ra}(\Phi_{j}),
\end{equation}

\noindent 2) For all $j,k=1,\ldots,M$,
\begin{eqnarray}\label{second moment}
\frac{\partial^2 r^{cy}(\alpha)}{\partial {\al_k}\partial {\al_j}}\lceil_{\alpha=0} & = & \beta_j \beta_k T^2 \sum_{n=0}^{\infty} \rho_{+}^{cy} \left(\varphi_j^{cy} \circ \Lj_{cy}^{*n} (\Phi_k^{cy} - \rho_{+}^{cy}(\Phi_{k}^{cy}) ) + \varphi_k^{cy} \circ  \Lj_{cy}^{*n} (\Phi_j^{cy} - \rho_{+}^{cy}(\Phi_{j}^{cy}) )\right) \nonumber \\
 & & + \delta_{k>j} \beta_j \beta_k T^2 \rho_{+}^{cy} \big( \cL_1^*\circ\cdots\circ\cL_{j-1}^*\circ\varphi_j\circ\cL_{j+1}^*\circ\cdots\circ\cL_{k-1}^*(\Phi_k) \big) \nonumber \\
 & & + \delta_{j>k} \beta_j \beta_k T^2 \rho_{+}^{cy} \big( \cL_1^*\circ\cdots\circ\cL_{k-1}^*\circ\varphi_k\circ\cL_{k+1}^*\circ\cdots\circ\cL_{j-1}^*(\Phi_j) \big) \nonumber \\
 & & +\delta_{jk} \beta_j^2T \rho_{+}^{cy} \left( \cL_1^*\circ\cdots\circ\cL_{j-1}^*\big(\cD_j(H_\cS',H_\cS')\big)\right).
\end{eqnarray}
and
\begin{eqnarray}\label{second momentrd}
\frac{\partial^2 r^{ra}(\alpha)}{\partial {\al_k}\partial {\al_j}}\lceil_{\alpha=0} & = & \beta_j \beta_k \left(\frac{T}{M}\right)^2 \sum_{n=0}^{\infty} \rho_{+}^{cy} \left(\varphi_j \circ \Lj_{ra}^{*n} \big(\Phi_k - \rho_{+}^{ra}(\Phi_{k}) \big) + \varphi_k \circ  \Lj_{ra}^{*n} \big(\Phi_j - \rho_{+}(\Phi_{j}) \big)\right) \nonumber \\
 & & +\delta_{jk} \beta_j^2\frac{T}{M} \rho_{+}^{ra} \big(\cD_j(H_\cS',H_\cS')\big).
\end{eqnarray}
\end{theo}

\begin{rema}\label{rem:fixedzeta} In the above theorem all the quantities depend on the vector $\bfzeta$ of thermal forces and all the results hold for any value of $\bfzeta$. We have not mentioned the dependence on $\bfzeta$ not to burden the notation.
\end{rema}

\begin{rema} The presence of the prefactors $T$ and $T^2$, resp. $\frac{T}{M}$ and $\frac{T^2}{M^2}$, in (\ref{first moment})-(\ref{second moment}), resp. (\ref{first moment}) and (\ref{second momentrd}), is due to the fact that the definitions of $r_n^\s(\alpha)$ correspond to the variation of entropy fluxes per interaction or cycle and not per unit time.
\end{rema}

\begin{rema} Of course (\ref{first moment}) is in agreement with (\ref{eq:fcsexpectation}).
\end{rema}

\begin{rema}\label{rem:gaspard} As mentioned at the end of the previous section (\ref{energy conservation}) is a direct consequence of the translation symmetry (\ref{eq:gaspard}) combined with (\ref{first moment}).
\end{rema}

Finally, we recall that the Green-Kubo formula (\ref{gkc})-(\ref{gkr}) can also be obtained via the moment generating function $\tilde r(\alpha)$, see e.g. \cite{JPR,JPW}. Indeed, we infer from (\ref{first moment}) and (\ref{eq:mgfconnection}) that 
\[
\frac{\partial^2 \tilde r_\bfzeta^{ra}(\alpha)}{\partial {\zeta_k}\partial {\al_j}}\lceil_{\alpha=\bfzeta=0} = \frac{T}{M}L_{jk}^{ra}, \quad \frac{\partial^2 \tilde r_\bfzeta^{cy}(\alpha)}{\partial {\zeta_k}\partial {\al_j}}\lceil_{\alpha=\bfzeta=0} = TL_{jk}^{cy} , \quad \frac{\partial^2 \tilde r_\bfzeta^{rcy}(\alpha)}{\partial {\zeta_k}\partial {\al_j}}\lceil_{\alpha=\bfzeta=0} = TL_{jk}^{rcy}. 
\]
If moreover (\nameref{Time-Reversal}) holds, then Evans-Searles symmetry gives for any $\alpha,\bfzeta$
\[
\tilde r_\bfzeta^{rcy}(\alpha) = \tilde r_\bfzeta^{cy} (\beta_{ref}{\bf 1}-\bfzeta-\alpha) , \quad \tilde r_\bfzeta^{ra} (\alpha) = \tilde r_\bfzeta^{ra} (\beta_{ref}{\bf 1}-\bfzeta-\alpha),
\]
and using translation symmetry with $\lambda=\beta_{ref}$ we get
\[
\tilde r_\bfzeta^{rcy}(\alpha) = \tilde r_\bfzeta^{cy} (-\bfzeta-\alpha), \quad \tilde r_\bfzeta^{ra} (\alpha) = \tilde r_\bfzeta^{ra} (-\bfzeta-\alpha).
\]
Using the chain rule we therefore have
\[
\frac{\partial^2 \tilde r_\bfzeta^{cy}(\alpha)}{\partial {\zeta_k}\partial {\al_j}}\lceil_{\alpha=\bfzeta=0} + \frac{\partial^2 \tilde r_\bfzeta^{rcy}(\alpha)}{\partial {\zeta_k}\partial {\al_j}}\lceil_{\alpha=\bfzeta=0}
= \frac{\partial^2 \tilde r_\bfzeta^{cy}(\alpha)}{\partial {\al_k}\partial {\al_j}}\lceil_{\alpha=\bfzeta=0}, \quad
\frac{\partial^2 \tilde r_\bfzeta^{ra}(\alpha)}{\partial {\zeta_k}\partial {\al_j}}\lceil_{\alpha=\bfzeta=0} = \frac{1}{2} \frac{\partial^2 \tilde r_\bfzeta^{ra}(\alpha)}{\partial {\al_k}\partial {\al_j}}\lceil_{\alpha=\bfzeta=0}. 
\]
As a summary
\begin{prop} If Assumptions (\nameref{ergodicity}), (\nameref{Non-Entanglement}) and (\nameref{Time-Reversal}) are satisfied then
\begin{equation*}\label{fd2}
\frac{\partial^2 \tilde r_\bfzeta^{cy}(\alpha)}{\partial {\al_k}\partial {\al_j}}\lceil_{\alpha=\bfzeta=0} = T(L_{jk}^{cy} + L_{jk}^{rcy}) = T(L_{jk}^{cy} + L_{kj}^{cy}) , \quad 
\frac{\partial^2 \tilde r_\bfzeta^{ra}(\alpha)}{\partial {\al_k}\partial {\al_j}}\lceil_{\alpha=\bfzeta=0} = \frac{2T}{M}L_{jk}^{ra},
\end{equation*}
where in the cyclic case the second equality follows from (\ref{ons}). 
\end{prop}
Note that similarly to what happens in quantum dynamical semigroups \cite{JPW}, using (\ref{second momentrd}) the above proposition gives another path to derive the Green-Kubo formula provided Assumption (\nameref{Time-Reversal}) holds. However due to the lack of global time-reversal invariance this does not allow us to retrieve (\ref{gkc}) even if (\nameref{Time-Reversal}) holds. 
Nevertheless we have the following relations which hold without (\nameref{Time-Reversal}).

\begin{coro}\label{corMoments} If Assumptions (\nameref{ergodicity}) and (\nameref{Non-Entanglement}) are satisfied, then
\begin{equation*}\label{fd1}
\frac{\partial^2 \tilde r_\bfzeta^{cy}(\alpha)}{\partial {\al_k}\partial {\al_j}}\lceil_{\alpha=\bfzeta=0} = T(L_{jk}^{cy} + L_{kj}^{cy}) , \quad 
\frac{\partial^2 \tilde r_\bfzeta^{ra}(\alpha)}{\partial {\al_k}\partial {\al_j}}\lceil_{\alpha=\bfzeta=0} = \frac{T}{M}(L_{jk}^{ra}+L_{kj}^{ra}).
\end{equation*}
\end{coro}

\begin{proof} Combining Remark \ref{rem:zerofluxeq}, Eq. (\ref{eq:revinteraction}) and (\ref{eq:fluxmap}), and the fact that $\rho_+$ is an invariant state for all the $\cL_j$'s, one easily gets that at equilibrium the various terms on the right-hand sides of (\ref{second moment})-(\ref{second momentrd}) coincide with those in (\ref{gkc})-(\ref{gkr}).
\end{proof}


\subsection{Fluctuation Theorem and Fluctuation Relation}\label{sec:fluctuationthm}

Let $e_{n,\rho}^{\s}(\al)$ and $e^{\s}(\al)$ denote the cumulant generating function of transient and large time entropy fluxes, i.e.  
\begin{equation*}
 e_{n,\rho}^{\s}(\al):=\frac{1}{\tau_{\s}}\log r_{n,\rho}^{\s}(\al) \ = \ \frac{1}{\tau_{\s}}\log  \mb{E}_{\P^\s} \left(e^{-\sum_{j=1}^M \al_j \mc{S}_{\Res_j}^{N}(\s)}\right), \ e^{\s}(\al):=\frac{1}{\tau_{\s}}\log r^{\s}(\al),
\end{equation*}
where $\tau_{cy}=T$ and $\tau_{ra}=\frac{T}{M}$, $\P^{cy}=\P$, $\P^{ra}=\hat\P$ and, with a slight abuse of notation, $\mc{S}_{\Res}^{N}(\s)$ stands for $\mc{S}_{\Res}^{N}(\bj)$ when $\s=ra$. Let also $\ds I_{\s}(\varsigma):= \sup_{\al \in \mb{R}^M}\left (\al .\varsigma - e^{\s}(-\al)\right)$ be the Fenchel-Legendre transform of $e^{\s}(-\al)$.

\begin{theo}\label{flucresults} Suppose Assumption (\nameref{ergodicity}) holds. Then 

\noindent 1) The entropy fluxes satisfy a large deviation principle with rate function $I_\s$. Namely, for any Borel set $G \subset \mb{R}^M$,
\begin{equation}\label{ld}
-\inf_{\varsigma\in\mathring{G}}I_{\s}(\varsigma)\leq\liminf_{N\rightarrow+\infty}\frac{1}{N\tau_\s}\log\mb{P}^\s \left(\frac{\mc{S}_{\Res}^{N}(\s)}{N\tau_\s}\in\mathring{G}\right) 
\leq\limsup_{N\rightarrow+\infty}\frac{1}{N\tau_\s}\log\mb{P}^\s\left(\frac{\mc{S}_{\Res}^{N}(\s)}{N\tau_\s} \in \mean{G}\right)\leq-\inf_{\varsigma\in\mean{G}}I_{\s}(\varsigma).
\end{equation}

\noindent 2) $\varsigma\mapsto I_{\s}(\varsigma)\in[0,+\infty]$ is closed convex, with compact level sets and $\ds \inf_{\varsigma\in\mb{R}^M}I_{\s}(\varsigma)=0$.

\noindent 3) The sequence of random vectors $\left(\frac{\mc{S}_{\Res}^{N}(\s)}{N\tau_{\s}}\right)_{N\in\N^*}$ converges in probability and exponentially fast to $\mc{S}_+(\s):=(-\beta_1 \rho_+^{\s}(\Phi_1^{\s}),\cdots,-\beta_M\rho_+^{\s}(\Phi_M^{\s}))$, that is
\begin{equation}\label{convergence}
\forall\epsilon>0,\exists A(\epsilon)>0,\forall N\in\mb{N}^*,\mb{P}^\s\left(\left\Vert\frac{\mc{S}_{\Res}^N(\s)}{N\tau_\s}-\mc{S}_+(\s) \right\Vert>\epsilon\right)\leq e^{-NA(\epsilon)}.
\end{equation}
Moreover $\frac{1}{\sqrt{N\tau_\s}}\left(\mc{S}_{\Res}^{N}(\s)-\mb{E}\left(\mc{S}_{\Res}^{N}(\s)\right)\right)$ converges in distribution towards a centered Gaussian $\mu$ whose covariance matrix is $(\dr{\al_j}\dr{\al_k}e^{\s}(0))_{j,k}$.

\noindent 4) If Assumption (\nameref{Time-Reversal}) holds, we have the Fluctuation Relation
\begin{equation}\label{flucrel}
\forall\varsigma\in\R^M,\ I_{cy}(-\varsigma)-I_{rcy}(\varsigma)=I_{ra}(-\varsigma)-I_{ra}(\varsigma)=\sum_{j=1}^M\varsigma_j,
\end{equation}
where $I_{rcy}$ is the analog of $I_{cy}$ for the reversed-order cyclic RIS.

\noindent 5) If Assumption (\nameref{Non-Entanglement}) holds then
\begin{equation}\label{energyconservation}
\forall\varsigma\in\R^M,\ \sum_{j = 1}^M\beta_j^{-1}\varsigma_j\neq 0\Rightarrow I_{\s}(\varsigma)=+\infty.
\end{equation}
In particular the gaussian measure $\mu$ in 3) is supported on the hyperplane $\ds \sum_{j = 1}^M\beta_j^{-1}\varsigma_j=0$.
\end{theo}

\begin{rema} Losely speaking the large deviation principle (\ref{ld}) can be written
\begin{equation*}
 \mb{P}^\s\left(\mc{S}_{\Res}^{N}(\s)=\varsigma\right)\approx e^{-N\tau_{\s} I_{\s}(\varsigma)}, 
\end{equation*}
as $N$ goes to $+\infty$. So equation (\ref{flucrel}) can also be translated as
\begin{equation*}
\frac{\mb{P}\left(\mc{S}_{\Res}^{N}(cy)=-\varsigma\right)}{\mb{P}\left(\mc{S}_{\Res}^{N}(pe)=\varsigma\right)} \approx e^{-NT(\varsigma_1 + \cdots + \varsigma_M)} \
\mbox{ and } \
\frac{\hat\P\left(\mc{S}_{\Res}^{N}(\bj)=-\varsigma\right)}{\hat\P\left(\mc{S}_{\Res}^{N}(\bj)=\varsigma\right)} \approx e^{-N\frac{T}{M}(\varsigma_1 + \cdots + \varsigma_M)}.
\end{equation*}
The latter is the original form of the Fluctuation Relation, see \cite{ES,GC}. 
\end{rema}

\begin{rema} Recall from Definition \ref{def:randomincrements} that $\mc{S}_{\Res}^{N}(\s)=\left( \mc{S}_{\Res_1}^{N}(\s),\ldots,\mc{S}_{\Res_M}^{N}(\s)\right)$ so that (\ref{convergence}) is of course related to Proposition \ref{prop:2ndlaw}. 
\end{rema}

\begin{rema} We would like to stress that except in point 5) the above theorem does not make use of Assumption (\nameref{Non-Entanglement}). Those results concern what happens \emph{far from equilibrium} hence there is indeed no reason that (\nameref{Non-Entanglement}), which is related to the notion of equilibrium, plays any role. 

Concerning 5), it is related to energy conservation (\ref{energy conservation}) which requires (\nameref{Non-Entanglement}) even when all temperatures are equal. Indeed, energy conservation is then equivalent to vanishing of entropy production which requires  (\nameref{Non-Entanglement}). 
\end{rema}
\medskip

Obviously Theorem \ref{flucresults} has its analog for the energy flux variables $\mc{Q}_{\Res}^{N}$. At equilibrium, 3. in the above theorem allows us to complete the results of Section \ref{lrt} about linear response with a Central Limit theorem on the large time behaviour of energy fluxes at equilibrium. Namely we have the following
\begin{theo}[Fluctuation-Dissipation]
\label{clt}
If Assumptions (\nameref{ergodicity}), (\nameref{Non-Entanglement}) and (\nameref{Time-Reversal}) hold, then the sequence of random vectors $\frac{1}{\sqrt{N\tau_{\s}}}\left(\mc{Q}_{\Res}^{N}(\s) - \mb{E} \left( \mc{Q}_{\Res}^{N}(\s) \right) \right)_{N\in\N^*}$ converges in distribution to a centered Gaussian whose covariance matrix $\left(D_{jk}^\s\right)_{1\leq j,k,\leq M}$ is given by $D_{jk}^{cy}= L_{jk}^{cy} + L_{jk}^{rcy}$ and $D_{jk}^{ra} = 2L_{jk}^{ra}$.
\end{theo}

\begin{proof} Theorem \ref{flucresults} gives the convergence with $D_{jk}^\s = \partial^2_{\al_j\al_k}\tilde e^{\s}(0)$ where $\tilde e^\s(\alpha) = \frac{1}{\tau_\s}\log \tilde r^\s(\alpha)$. The result follows using Corollary \ref{corMoments} and that at equilibrium $\dr{\al_j}\tilde r(0)= \tau_\s\rho_+^\s(\Phi_j^\s)=0$.
\end{proof}


\section{Proof of the main results}\label{proofs}


\subsection{Proof of Theorem \ref{gko}}\label{prooflrt}

This proof is inspired by the one given in \cite{LS} (see also \cite{JPW}) for open quantum systems interacting with thermal reservoirs in the Van Hove weak coupling limit. In order to alleviate the notation, all the quantities without any $\zeta$ or $\bfzeta$ parameter should be understood at equilibrium $\bfzeta=0$.

\noindent 1)-2) The differentiability of $\rho_{+,\bfzeta}^\s$ follows from Proposition \ref{erprop} and the one of $\Phi^\s_{j,\bfzeta}$ is clear from its definition. Thus we have
\begin{equation}\label{pf:kinetic}
L_{jk}^\s = \partial_{\zeta_k} \rho^\s_{+,\bfzeta}\big\lceil_{\bfzeta = 0} (\Phi^\s_j) + \rho_{+} \left(\partial_{\zeta_k}\Phi^\s_{j,\bfzeta}\big\lceil_{\bfzeta = 0}\right). 
\end{equation}
Since $\cL_{\s,\bfzeta}\left(\rho^\s_{+,\bfzeta}\right)=\rho^\s_{+,\bfzeta}$ for any $\bfzeta$ we get
$\ds \left(\un - \Lj_{\s,\bfzeta}\right)\left(\dr{\zeta_k} \rho^\s_{+,\bfzeta}\right) = \left(\dr{\zeta_k}\Lj_{\s,\bfzeta}\right)(\rho^\s_{+,\bfzeta})$,
hence
\begin{equation}\label{pregk}
\left(\un - \Lj_{\s,\bfzeta}^N\right)\left(\dr{\zeta_k} \rho^{\s}_{+,\bfzeta} \right) = \sum_{n=0}^{N-1} \Lj_{\s,\bfzeta}^n \circ \dr{\zeta_k}\Lj_{\s,\bfzeta} (\rho^\s_{+,\bfzeta}),\qquad \forall N\geq 1.
\end{equation}
Using Proposition \ref{erprop} we have
$\ds
\lim_{N\to\infty} \Lj_{\s,\bfzeta}^N\left(\dr{\zeta_k} \rho^{\s}_{+,\bfzeta} \right) = \tr \left(\dr{\zeta_k} \rho^{\s}_{+,\bfzeta} \right) \rho_{+,\bfzeta}^\s = 0
$
where we have used that $\tr(\rho_{+,\bfzeta}^\s)=1$ for all $\bfzeta$ so that $\tr\left(\dr{\zeta_k} \rho^\s_{+,\bfzeta}\right)=0$. Letting $N\to\infty$ in (\ref{pregk}) we get
$\ds \dr{\zeta_k} \rho^{\s}_{+,\bfzeta} \lceil_{\bfzeta = 0} = \sum_{n=0}^{+\infty} \Lj_{\s}^n \circ  \dr{\zeta_k} \Lj_{\s,\bfzeta}\lceil_{\bfzeta = 0} (\rho_{+}),$ and (\ref{pf:kinetic}) becomes
\begin{equation}\label{pf:kinetic2}
L_{jk}^\s =  \sum_{n=0}^{+\infty} \rho_{+} \left(\dr{\zeta_k} \Lj_{\s,\bfzeta}^*\lceil_{\bfzeta = 0} \, \circ \, \Lj_{\s}^{*n} \, (\Phi^\s_{j})\right) + \rho_{+} \left(\partial_{\zeta_k}\Phi^\s_{j,\bfzeta}\big\lceil_{\bfzeta = 0}\right). 
\end{equation}

\begin{lemm}\label{lem:rdmderiv} For all $k$ and $X\in\cB(\cH_\cS)$ one has $\ds \rho_{+} \left(\dr{\zeta_k} \Lj_{k,\zeta_k}^*\lceil_{\zeta_k = 0} (X)\right)= T \rho_{+}\left( X \Phi_{k,\rev}\right).$ 
\end{lemm}

\begin{proof} Using $\dr{\zeta_k} \rho_{\cE_k,\zeta_k} = H_{\cE_k}\rho_{\cE_k,\zeta_k} - \tr(H_{\cE_k}\rho_{\cE_k,\zeta_k})\rho_{\cE_k,\zeta_k}$ and the definition of $\Lj_{k,\zeta_k}$ we have
\begin{eqnarray*}
\lefteqn{ \rho_{+} \left(\dr{\zeta_k} \Lj_{k,\zeta_k}^*\lceil_{\zeta_k = 0} (X)\right)} \\
 & = & \tr\left( \un \otimes H_{\cE_k} \times \rho_{+} \otimes \rho_{\cE_k} \, U_k^* \, X\otimes \un \, U_k \right) - \tr\left(\rho_{+} \otimes \rho_{\cE_k} \, U_k^* \, X\otimes \un \, U_k \right) \times \tr(H_{\cE_k}\rho_{\cE_k})\\
 & = & \tr\left( U_k\,  \un \otimes H_{\cE_k} \, U_k^* \times \rho_{+}X \otimes \rho_{\cE_k} \right) - \tr\left(\rho_{+} X \right) \times \tr(H_{\cE_k}\rho_{\cE_k})\\
 & = & \tr\left( \big(U_k\,  \un \otimes H_{\cE_k} \, U_k^* - \un \otimes H_{\cE_k}\big)  \times \rho_{+}X \otimes \rho_{\cE_k} \right) \\
 & = & T \tr \left( \Phi_{k,\rev} \rho_{+}X\right),
\end{eqnarray*}
where we have used the cyclicity of the trace and Assumption (\nameref{Non-Entanglement}) in line 2.
\end{proof}

\begin{coro} For all $k=1,\ldots,M$ and $X\in\cB(\cH_\cS)$ one has 
\begin{align}
\rho_{+} \left(\dr{\zeta_k} \Lj_{cy,\bfzeta}^*\lceil_{\bfzeta = 0} (X)\right) = & \ T \rho_{+}\left( \cL_{k+1}^*\circ \cdots \circ \cL_{M}^*(X) \Phi_{k,\rev}\right) \ = \ T \rho_{+}\left( X \Phi_{k,\rev}^{cy}\right), \label{gkcsum}\\
\rho_{+} \left(\dr{\zeta_k} \Lj_{ra,\bfzeta}^*\lceil_{\bfzeta = 0} (X)\right) = & \ \frac{T}{M} \rho_{+}\left( X \Phi_{k,\rev}\right). \label{gkrsum}
\end{align}
\end{coro}

\begin{proof} From (\ref{def:cycle}) we infer that $\ds  \dr{\zeta_k}\Lj_{cy,\bfzeta}^*\lceil_{\bfzeta=0}= \cL_{1}^*\circ\cdots \circ \cL_{k-1}^*\circ \dr{\zeta_k} \Lj_{k,\zeta_k}^*\lceil_{\zeta_k = 0} \circ \cL_{k+1}^*\circ \cdots \circ \cL_{M}^*$. The first equality in (\ref{gkcsum}) then follows from Lemma \ref{lem:rdmderiv} and the fact that $\rho_{+}$ is a joint invariant state for $\cL_{j}$, $j=1,\ldots,M$, and the second equality from (\ref{eq:revinteraction})  and Definition \ref{def:fluxrev}.

The second identity is immediate by definition of $\Lj_{ra,\bfzeta}$.
\end{proof}

Inserting (\ref{gkcsum})-(\ref{gkrsum}) into (\ref{pf:kinetic2}) leads to the infinite sums in (\ref{gkc})-(\ref{gkr})-(\ref{gkc2}). It thus remains to compute $\ds \rho_{+} \left(\partial_{\zeta_k}\Phi^\s_{j,\bfzeta}\big\lceil_{\bfzeta = 0}\right)$. For $\s=ra$ it follows from (\ref{translation2}) and Lemma \ref{lem:rdmderiv} that 
\[
\rho_{+} \left(\partial_{\zeta_k}\Phi_{j,\bfzeta}\big\lceil_{\bfzeta = 0}\right) = \delta_{jk}\frac{1}{T}\, \rho_+ \left( \partial_{\zeta_j}\cL_{j,\zeta_j}^*\lceil_{\zeta_j=0}(H_\cS')\right) = \delta_{jk} \rho_+\left( H_\cS' \Phi_{j,\rev}\right), 
\]
while for $\s=cy$, using moreover (\ref{def:cyclicflux}) and the fact that $\rho_{+}$ is a joint invariant state for the $\cL_{j}$'s, we have for $k<j$ 
\begin{eqnarray*}
\rho_{+} \left(\partial_{\zeta_k}\Phi_{j,\bfzeta}^{cy}\big\lceil_{\bfzeta = 0}\right) & = &  \rho_+\left(\partial_{\zeta_k}\cL_{k,\zeta_k}^*\lceil_{\zeta_k=0}\, \circ \, \cL_{k+1}^*\circ \cdots\circ \cL_{j-1}^* (\Phi_j) \right) \\
 & = & T\, \rho_+\left(\cL_{k+1}^*\circ \cdots\circ \cL_{j-1}^* (\Phi_j) \times \Phi_{k,\rev} \right),
\end{eqnarray*}
while $\ds \rho_{+} \left(\partial_{\zeta_k}\Phi_{j,\bfzeta}^{cy}\big\lceil_{\bfzeta = 0}\right)=0$ for $k>j$ and
\[
\rho_{+} \left(\partial_{\zeta_j}\Phi_{j,\bfzeta}^{cy}\big\lceil_{\bfzeta = 0}\right) = \frac{1}{T}\, \rho_+ \left( \partial_{\zeta_j}\cL_{j,\zeta_j}^*\lceil_{\zeta_j=0}(H_\cS')\right) = \rho_+\left( H_\cS' \Phi_{j,\rev}\right).
\]
Equations (\ref{gkc})-(\ref{gkr}) finally follow because $\ds \rho_+\left( H_\cS' \Phi_{j,\rev}\right) = \frac{1}{2} \rho_+(\cD_j(H_\Sy ',H_\cS'))$. Indeed, 
\begin{eqnarray*}
\rho_+\left( H_\cS' \Phi_{j,\rev} \right) & = & \frac{1}{T} \,\rho_+\Big( H_\cS'^2 - H_\cS' \Lj_{j,\rev}^* (H_{\Sy}')\Big) \\
 & = & \frac{1}{T}\, \rho_+\Big( H_\cS'^2 - \Lj_{j}^* (H_{\Sy}')H_\cS'  \Big) \\
 & = & \frac{1}{2T}\, \rho_+\Big( H_\cS'^2 + \Lj_j^*(H_\cS'^2)- \Lj_{j}^* (H_{\Sy}')H_\cS' - H_\cS' \Lj_{j}^* (H_{\Sy}')  \Big) \\
 & = & \frac{1}{2}\, \rho_+\big(\cD_j(H_\Sy ',H_\cS')\big),
\end{eqnarray*}
where we have used (\ref{translationrev}) in line 1, (\ref{eq:revinteraction}) in line 2, and that at equilibrium $\rho_+$ is $\cL_j$ invariant and $[\cL_j^*(H_\cS'),H_\cS']=0$, see Remark \ref{rem:effecthaminvariance},
in line 3.

\medskip

3) We now prove the Onsager relations (\ref{ons}). We thus now assume that Assumption (\nameref{Time-Reversal}) also holds. It follows from (\ref{eq:revLj}) that
$\ds
\Theta\,\circ\,\cL_{ra}^* \,\circ \,\Theta = \cL_{ra}^{*\rho_+}$, and $\rho_+(\Theta(X))=\rho_+(X^*)$ for any $X\in\cB(\cH_\cS)$ because $\rho_+$ is $\Theta$ invariant. Hence using (\ref{eq:phireversal}) we get for all $n$
\begin{eqnarray}\label{onsproof}
 \rho_+ (\Lj_{ra}^{*n} (\Phi_j) \Phi_{k,\rev}) & = & \rho_+ \Big(\Theta (\Phi_{k,\rev}) \times (\Lj_{ra}^{*\rho_+})^n(\Theta(\Phi_j))\Big) \nonumber\\
 & = & \rho_+ \Big(\Phi_k \times (\Lj_{ra}^{*\rho_+})^n(\Phi_{j,\rev})\Big) \nonumber\\
 & = & \rho_+ (\Lj_{ra}^{*n} (\Phi_k) \Phi_{j,\rev}),
\end{eqnarray}
from which $L_{jk}^{ra}=L_{kj}^{ra}$ follows.

Similarly, let
\begin{equation*}
 \Lj_{rcy}:=\Lj_1\circ\Lj_2\circ\cdots\circ\Lj_M,\ \  \Phi_j^{rcy}=\frac{1}{T}\Lj_M^*\circ\cdots\circ\Lj_{j+1}^*(\Phi_j) \ \mbox{ and } \ \Phi_{j,\rev}^{rcy}=\frac{1}{T}\Lj_{1,\rev}^{*}\circ\cdots\circ\Lj_{j-1,\rev}^{*}(\Phi_{j,\rev}),
\end{equation*}
denote the reversed-order analogs of $\Lj_{cy}$, $\Phi_j^{cy}$ and $\Phi_{j,\rev}^{cy}$. The associated kinetic coefficients $L_{jk}^{rcy}$ are thus given by
\begin{equation*}\label{gkp}
\begin{aligned}
L_{jk}^{rcy} \ = \ & T \sum_{n=0}^{+\infty} \rho_+ (\Lj_{rcy}^{*n} (\Phi_j^{rcy}) \Phi_{k,\rev}^{rcy}) \\
& + \delta_{k>j}\, T\rho_+\Big( \Lj_{k-1}^*\circ\cdots\circ\Lj_{j+1}^*(\Phi_j) \times \Phi_{k,\rev}\Big) +\frac{1}{2}\delta_{jk} \rho_+(\cD_j(H_\Sy ',H_\cS')),
\end{aligned}
\end{equation*}
It follows from Lemma \ref{eq:revLj} and (\ref{eq:rdmtri})-(\ref{eq:phireversal}) that 
\[
\Theta\circ \cL_{cy}^*\circ \Theta = \cL_{rcy}^{*\rho_+}, \ \ \Theta\left( \Phi_{j}^{cy} \right) =  -\Phi_{j,\rev}^{rcy} \ \mbox{ and } \ \Theta\left( \Phi_{k,\rev}^{cy} \right) =  -\Phi_{k}^{rcy},
\]
and the same reasoning as in (\ref{onsproof}) shows that $L_{jk}^{cy}=L_{kj}^{rcy}$.


\subsection{Proof of Proposition \ref{prop:mgfdeformedrdm}}\label{sec:proofmgfdeformedrdm}

Since $[H_{\cE_j},H_{\cE_k}]=[H_{\cE_j},H_k]=0$ for $j\neq k$, we have
\[
\mc{U}_N^\bfs (\bj) = \Pi_{s_N'}\left(S_{\En_{j_N}}^N\right)\cdots \Pi_{s_1'}\left(S_{\En_{j_1}}^1\right) \times U_N\cdots U_1\times  \Pi_{s_1}\left(S_{\En_{j_1}}^1\right)\cdots \Pi_{s_N}\left(S_{\En_{j_N}}^N\right),
\]
so that (\ref{eq:entropyincrementproba}) becomes
\begin{eqnarray*}
\lefteqn{ \mb{P} (\mc{S}_{\Res}^{N}(\bs{j})=\bs\varsigma) } \\
 & = \!\! & \sum \tr \left(  \Pi_{s_N'}(S_{\En_{j_N}}^N)\cdots \Pi_{s_1'}(S_{\En_{j_1}}^1) \times U_N\cdots U_1\times  \Pi_{s_1}(S_{\En_{j_1}}^1)\cdots \Pi_{s_N}(S_{\En_{j_N}}) \times \rho_{tot}^N(\bs{j}) \times U_1^*\cdots U_N^*  \right),
\end{eqnarray*}
where the sum is as in (\ref{eq:entropyincrementproba}). Hence
\begin{eqnarray*}\label{specmgf:proof}
\lefteqn{r_{N,\rho}^{\bj} (\al) }\nonumber\\
 & = & \sum_{\varsigma_1,\ldots,\varsigma_M} \e^{-(\alpha_1\varsigma_1+\cdots+\alpha_M\varsigma_M)} \mb{P} (\mc{S}_{\Res}^{N}(\bs{j})=\bs\varsigma)\nonumber\\
 & = & \sum_{\bfs}  \prod_{n=1}^N \e^{-\alpha_{j_n}(s'_n-s_n)}  \tr\left(  \Pi_{s_N'}(S_{\En_{j_N}}^N)\cdots \Pi_{s_1'}(S_{\En_{j_1}}^1) \times U_N\cdots U_1\times  \Pi_{s_1}(S_{\En_{j_1}}^1)\cdots \Pi_{s_N}(S_{\En_{j_N}}) \right.\nonumber\\
 & & \qquad\qquad\qquad\qquad\qquad\qquad\qquad\qquad\qquad\qquad\qquad\qquad\qquad\qquad \times \rho_{tot}^N(\bs{j}) \times U_1^*\cdots U_N^*   \Big)\nonumber\\
 & = & \tr\left( \e^{-\alpha_{j_N}S_{\En_{j_N}}^N}\cdots\e^{-\alpha_{j_1}S_{\En_{j_1}}^1} \times U_N\cdots U_1\times \e^{\alpha_{j_1}S_{\En_{j_1}}^1} \cdots \e^{\alpha_{j_N}S_{\En_{j_N}}^N} \times \rho_{tot}^N(\bs{j}) \times U_1^*\cdots U_N^* \right). 
\end{eqnarray*}
Now, recall that $\rho_{tot}^N(\bs{j})= \rho\otimes\rho_{\En_{j_1}}\otimes\cdots\otimes\rho_{\En_{j_N}}$ where $\ds \rho_{\En_{j_n}} = \frac{\e^{-S_{\En_{j_n}}}}{\tr\left( \e^{-S_{\En_{j_n}}}\right)}$ for any $n$, and that the $\e^{-S_{\En_{j_n}}^n}$'s act on different probes. Hence we get
\[
r_{N,\rho}^{\bj} (\al)  =  \tr\left( \un_\cS\otimes \rho_{\En_{j_1}}^{\alpha_{j_1}}\otimes\cdots\otimes\rho_{\En_{j_N}}^{\alpha_{j_N}} \times U_N\cdots U_1\times \rho\otimes\rho_{\En_{j_1}}^{1-\alpha_{j_1}}\otimes\cdots\otimes\rho_{\En_{j_N}}^{1-\alpha_{j_N}}  \times U_1^*\cdots U_N^*   \right), 
\]
and (\ref{specmgf}) follows exactly in the same way as (\ref{rismarkovianform}) follows from (\ref{rishamiltonianform}).


\subsection{Proof of Theorem \ref{theojpw}}

1) Since $\Lj_{\s,\bfzeta}^{[\al]*}$ is completely positive, and primitive  by Proposition \ref{deformedsp}, it follows from Perron-Frobenius theory for completely positive maps \cite{EHK} that $r_\bfzeta^\s(\alpha)$ is a simple dominant eigenvalue with positive definite left and right eigenvectors. Then (\ref{eq:largetimemgf}) follows from (\ref{specmgfcyra}) by a standard argument: there exists $\gamma>0$ (spectral gap) such that for all $n$ one has $\left(\Lj_{\s}^{[\al]^*}\right)^n = (r_\bfzeta^\s(\alpha))^n| A \rangle\langle \nu| +O\left((r_\bfzeta^\s(\alpha)-\gamma)^n\right)$ where $\nu$ and $A$ denote the positive definite left and right eigenvectors of $\Lj_{\s,\bfzeta}^{[\al]*}$ normalized such that $\nu(A)=1$. Thus
\[
 r_{n,\rho}^{\s} (\al) = \rho\left( \left(\Lj_{\s}^{[\al]^*}\right)^n (\un) \right) = (r_\bfzeta^\s(\alpha))^n \left[ \rho(A)\times \nu(\un) + O\left(\left(1-\frac{\gamma}{r_\bfzeta^\s(\alpha)}\right)^n\right)\right],
\]
and the result follows since both $\rho(A)>0$ and $\nu(\un)>0$.

\smallskip

\noindent 2) The symmetry relies on the following Lemma which is a direct consequence of (\nameref{Time-Reversal}).
\begin{lemm} For any $j=1,\ldots,M$ and $\bfzeta,\alpha\in\R^M$ one has $\Theta\circ\cL_{j,\bfzeta}^{[\al]*}\circ\Theta = \cL_{j,\bfzeta}^{[{\bf 1}-\alpha]}$. As a consequence
\[
\Theta \circ \Lj_{cy,\bfzeta}^{[\al]*} \circ \Theta \ = \ \cL_{rcy,\bfzeta}^{[{\bf 1}-\alpha]} \quad \mbox{ and } \quad \Theta \circ \Lj_{ra,\bfzeta}^{[\al]*} \circ \Theta = \Lj_{ra,\bfzeta}^{[{\bf 1}-\al]}.
\]
\end{lemm}
\noindent The lemma indeed implies that the maps $\ds \Lj_{ra,\bfzeta}^{[\al]*}$ and $\ds \Lj_{ra,\bfzeta}^{[{\bf 1}-\al]}$, resp. $\Lj_{cy,\bfzeta}^{[\al]*}$ and $\cL_{rcy,\bfzeta}^{[{\bf 1}-\alpha]}$, have the same spectral radius which proves (\ref{ES}).

\begin{proof} Denote by $\Theta_j$ the time-reversal of $\cE_j$ as in Section \ref{ssec:tri}. Then for any $X,Y$ we have
\begin{eqnarray*}
\lefteqn{\tr\left( Y\times \Theta\circ\cL_j^{[\al]*}\circ\Theta (X) \right)} \\
 & = & \overline{\tr\left( \Theta(Y) \times \cL_j^{[\al]*}\left(\Theta(X)\right) \right)} \\
 & = & \overline{\tr\left( \Theta(Y)\otimes \rho_{\cE_j}^{1-\alpha_j} \times U_j^* \times \Theta(X)\otimes \rho_{\cE_j}^{\alpha_j} \times U_j\right)} \\
 & = & \tr\left( Y\otimes \Theta_j(\rho_{\cE_j}^{1-\alpha_j}) \times \Theta\otimes\Theta_j(U_j^*) \times X\otimes \Theta_j(\rho_{\cE_j}^{\alpha_j}) \times \Theta\otimes\Theta_j(U_j)\right) \\
 & = & \tr\left( Y\otimes \rho_{\cE_j}^{1-\alpha_j} \times U_j \times X\otimes \rho_{\cE_j}^{\alpha_j} \times U_j^*\right) \\
 & = & \tr\left( Y \times \cL_j^{[{\bf 1}-\al]}(X) \right),
\end{eqnarray*}
where we have used Assumption (\nameref{Time-Reversal}) in the $4$-th line.
\end{proof}

\noindent 3) The argument is again of isospectral type and relies on the following lemma. 

\begin{lemm}\label{lem:Lalphaderiv} If (\nameref{Non-Entanglement}) holds then for any $j=1,\ldots,M$, $\alpha=(\alpha_1,\ldots,\alpha_M)\in\R^M$ and $\bfzeta=(\zeta_1,\ldots,\zeta_M)\in\R^M$ and $X\in\cB(\cH_\cS)$ one has
\begin{equation}\label{eq:Lalphaequiv}
\cL_j^{[\alpha]*}(X) = \cL_j^*\left( X \, \e^{\beta_j\alpha_j H_\cS'} \right)\e^{-\beta_j\alpha_j H_\cS'} = \e^{-\beta_j\alpha_j H_\cS'} \cL_j^*\left( \e^{\beta_j\alpha_j H_\cS'} \, X \right),
\end{equation} 
where $\beta_j$ is the inverse temperature of $\cE_j$, i.e. such that $\zeta_j=\beta_{ref}-\beta_j$. 
\end{lemm}

\begin{rema} Note that the mere existence of $H_\cS'$ requires (\nameref{Non-Entanglement}).
\end{rema}

\begin{proof} For any $X\in\cB(\cH_\cS)$ we have
\begin{eqnarray*}
\cL_j^{[\alpha]*}(X) & = & \tr_{\En_j} \left(\un \otimes \e^{\beta_j \alpha_j H_{\cE_j}} \rho_{\En_j} \times U_j^* \times X \otimes \e^{-\beta_j \alpha_j H_{\cE_j}}  \times U_j\right) \\
& = & \tr_{\En_j} \left(\un\otimes \rho_{\En_j} \times U_j^* \times X\e^{\beta_j \alpha_j H_{\cS}'} \otimes \un \times \e^{-\beta_j \alpha_j (H_\cS'+H_{\cE_j})}  \times U_j \times \un \otimes \e^{\beta_j \alpha_j H_{\cE_j}}\right) \\
 &= & \tr_{\En_j} \left(\un\otimes \rho_{\En_j} \times U_j^* \times X\e^{\beta_j \alpha_j H_{\cS}'} \otimes \un \times U_j\times  \e^{-\beta_j \alpha_j H_\cS'} \otimes \un  \right)\\
 & = & \cL_j^*\left( X \, \e^{\beta_j\alpha_j H_\cS'} \right)\e^{-\beta_j\alpha_j H_\cS'},
\end{eqnarray*} 
where we have used (\ref{conservation}) in line 3. The second equality is a direct consequence of (\ref{eq:gauge}).
\end{proof}

We finally deduce (\ref{ts}). For $\lambda\in \R$ let $\cM_\lambda$ denote the right-multiplication by $\e^{\lambda H_\cS'}$ on $\cB(\cH_\cS)$. It follows from Lemma \ref{lem:Lalphaderiv} that, for any $j,\alpha,\bfzeta$, we have
$\ds 
\cL_{j,\bfzeta}^{[\alpha+\lambda\beta^{-1}]*} = \cM_\lambda^{-1} \circ \cL_{j,\bfzeta}^{[\alpha]*} \circ \cM_\lambda,
$
so that 
\[
\cL_{\s,\bfzeta}^{[\alpha+\lambda\beta^{-1}]*} = \cM_\lambda^{-1} \circ \cL_{\s,\bfzeta}^{[\alpha]*} \circ \cM_\lambda, \quad \s=cy,\, ra.
\]
The maps $\ds \cL_{\s,\bfzeta}^{[\alpha+\lambda\beta^{-1}]*}$ and $\ds \cL_{\s,\bfzeta}^{[\alpha]*}$ thus have the same spectral radius, i.e. $\ds r^\s_\bfzeta(\alpha) = r^\s_\bfzeta(\alpha+\lambda \beta^{-1})$.


\subsection{Proof of Theorem \ref{momentstheo}}\label{proofmoments}

The proof is an adaptation of the one in \cite{JPW} for quantum dynamical semigroups. In all this section the thermodynamical parameter $\bfzeta$ is arbitrary but fixed, see Remark \ref{rem:fixedzeta}, and we shall omit it.

For $\alpha=0$ the operator $\cL_\s^{[0]*}=\cL_\s^*$ has a simple dominant eigenvalue $1$. By perturbation theory there exists a small circle $\Gamma$ centered at $1$ such that  for $\alpha\in\R^M$ sufficiently close to $0$ the only point in $\sp\left(\cL_\s^{[\al]*}\right)$ inside or on $\Gamma$ is its dominant eigenvalue $r_\s(\alpha)$. We shall further denote by $P_\s^{[\al]}$ its associated eigenprojection. Note that $P_\s^{[0]} (X) = \rho_+^\s(X)\un$.

For $n=0,1$ denote by $E_n^{\s} (\al)$ the quantity
\begin{equation}\label{def:En}
E_n^{\s} (\al) := \oint_{\Gamma} (z-1)^n  \rho_{+}^\s \left(\left(z-\Lj_{\s}^{[\al]*}\right)^{-1}(\un)\right) \frac{{\rm d} z}{2\pi i}.
\end{equation}
Writing $E_n^\s(\al)$ as
\[
E_n^{\s} (\al) := \oint_{\Gamma} (z-1)^n  \rho_{+}^\s \left(\left(z-\Lj_{\s}^{[\al]*}\right)^{-1}\circ (Id-P_\s^{[\al]})(\un)\right) \frac{{\rm d} z}{2\pi i} + \oint_{\Gamma} \frac{(z-1)^n}{z-r_\s(\alpha)}  \rho_{+}^\s \left(P_\s^{[\al]}(\un)\right) \frac{{\rm d} z}{2\pi i},
\]
the first term on the right-hand side is analytic inside $\Gamma$ and it follows from Cauchy's integral formula that
\begin{equation}\label{eq:Enformulas}
E_0^{\s}(\al) = \rho_{+}^\s \left(P_\s^{[\al]}(\un)\right) \ \mbox{ and } \ E_1^{\s}(\al) =(r^{\s}(\al)-1) \rho_{+}^\s \left(P_\s^{[\al]}(\un)\right),
\end{equation}
so that
\begin{equation}\label{eq:linkEnr}
r^{\s}(\al) = 1+ \frac{E_1^{\s}(\al)}{E_0^{\s}(\al)}.
\end{equation}
Since $P_\s^{[0]} (\un) = \un$ and $r^\s(0)=1$ it follows from (\ref{eq:Enformulas}) that $E_0^{\s}(0)=1$, $E_1^{\s}(0)=0$ hence $\ds \dr{\al_j} r^\s(0)=\ds \dr{\al_j} E_1^\s(0)$.
Using (\ref{def:En}) we get 
\begin{equation}\label{eq:Enderiv1}
\ds \dr{\al_j} E_n^\s(\alpha)  =  \oint_{\Gamma} (z-1)^n  \rho_{+}^\s \left(\left(z-\Lj_{\s}^{[\al]*}\right)^{-1} \circ \dr{\al_j} \Lj_{\s}^{[\al]*}\, \circ \left(z-\Lj_{\s}^{[\al]*}\right)^{-1}(\un)\right) \frac{{\rm d} z}{2\pi i},
\end{equation}
so that, because $\cL_\s^*(\un)=\un$ and $\cL_\s(\rho_+^\s)=\rho_+^\s$,
\[
\ds \dr{\al_j} E_n^\s(0) = \oint_{\Gamma} (z-1)^{n-2}  \rho_{+}^\s \left( \dr{\al_j} \Lj_{\s}^{[\al]*}\lceil_{\alpha=0}(\un)\right) \frac{{\rm d} z}{2\pi i},
\]
hence
\begin{equation}\label{eq:E1deriv}
\dr{\al_j} E_0^\s(0)=0 \quad \mbox{and} \quad \dr{\al_j} E_1^\s(0) = \rho_{+}^\s \left( \dr{\al_j} \Lj_{\s}^{[\al]*}\lceil_{\alpha=0}(\un)\right).
\end{equation}
Eq. (\ref{eq:linkEnr}) and (\ref{eq:E1deriv}) then give $\ds \dr{\al_k}\dr{\al_j} r^\s(0)=\ds \dr{\al_k}\dr{\al_j} E_1^\s(0)$, while from (\ref{eq:Enderiv1}) we infer that
\begin{eqnarray}\label{eq:Enderivative2}
\dr{\al_k}\dr{\al_j}E_1^\s (0)  & = & \oint_{\Gamma} \, (z-1)^{-1} \rho_{+}^\s \left(\dr{\al_k} \Lj_{\s}^{[\al]*}\lceil_{\al=0}\, \circ \left(z-\Lj_{\s}^*\right)^{-1} \circ \dr{\al_j} \Lj_{\s}^{[\al]*}\lceil_{\al=0}(\un)\right)  \frac{{\rm d} z}{2\pi i} \nonumber \\
 & & + \oint_{\Gamma} \, (z-1)^{-1} \rho_{+}^\s \left(\dr{\al_j} \Lj_{\s}^{[\al]*}\lceil_{\al=0}\, \circ \left(z-\Lj_{\s}^*\right)^{-1} \circ \dr{\al_k} \Lj_{\s}^{[\al]*}\lceil_{\al=0}(\un)\right)  \frac{{\rm d} z}{2\pi i} \nonumber \\
& &+  \oint_{\Gamma} \, (z-1)^{-1} \rho_{+}^\s \left(\dr{\al_j}\dr{\al_k} \Lj_{\s}^{[\al]*}\lceil_{\al=0}(\un)\right) \frac{{\rm d} z}{2\pi i}\\
& =: & I+ II + III \nonumber
\end{eqnarray}

\noindent The next lemma is the main technical ingredient leading from (\ref{eq:E1deriv})-(\ref{eq:Enderivative2}) to (\ref{first moment})-(\ref{second momentrd}).

\begin{lemm} For any $j,k=1,\ldots,M$ one has
\begin{equation}\label{eq:Lalpharacyderiv1}
\dr{\al_j}\cL_{ra}^{[\al]*}\lceil_{\al=0}(X) = \beta_j \frac{T}{M} \varphi_j(X), \quad \dr{\al_j}\cL_{cy}^{[\al]*}\lceil_{\al=0}(X) = \beta_j T \varphi_j^{cy}(X),
\end{equation}
\begin{equation}\label{eq:Lalpharaderiv2}
\dr{\al_k}\dr{\al_j}\cL_{ra}^{[\al]*}\lceil_{\al=0}(\un) = \delta_{jk}\beta_j^2\frac{T}{M} \cD_j(H_\cS',H_\cS'),
\end{equation}
and
\begin{equation}\label{eq:Lalphacyderiv2}
\dr{\al_k}\dr{\al_j}\cL_{cy}^{[\al]*}\lceil_{\al=0}(\un) = \left\{  \begin{array}{ll} 
\beta_j \beta_k T^2 \,\cL_1^*\circ\cdots\circ\cL_{j-1}^*\circ\varphi_j\circ\cL_{j+1}^*\circ\cdots\circ\cL_{k-1}^*(\Phi_k),  & \mbox{if } k>j,  \\
\beta_j \beta_k T^2  \,\cL_1^*\circ\cdots\circ\cL_{k-1}^*\circ\varphi_k\circ\cL_{k+1}^*\circ\cdots\circ\cL_{j-1}^*(\Phi_j), & \mbox{if } k<j, \\
\beta_j^2T  \, \cL_1^*\circ\cdots\circ\cL_{j-1}^*\big(\cD_j(H_\cS',H_\cS')\big), & \mbox{if } j=k. \end{array} \right.
\end{equation}
\end{lemm}

\begin{proof} Using Lemma \ref{lem:Lalphaderiv}, by straightforward calculation we get 
\[
\dr{\al_j} \cL_j^{[\alpha]*}(X) = \beta_j \left( \cL_j^{[\alpha]*}(XH_\cS') -\cL_j^{[\alpha]*}(X)H_\cS' \right)
\] 
and 
\[
\dr{\al_k} \dr{\al_j} \cL_j^{[\alpha]*}(X) = \delta_{jk} \beta_j^2 \left( \cL_j^{[\alpha]*}(XH_\cS'^2) -2\cL_j^{[\alpha]*}(XH_\cS')H_\cS' + \cL_j^{[\alpha]*}(X)H_\cS'^2\right).
\]
Using the definitions (\ref{def:fluxmap}) of $\varphi_j(X)$ and (\ref{def:dissipation}) of $\cD_j(X,Y)$ we thus obtain
\begin{equation}\label{eq:Lalphaderiv}
\dr{\al_j} \cL_j^{[\alpha]*}\lceil_{\al=0}(X) = \beta_jT\varphi_j(X) \quad \mbox{and} \quad \dr{\al_k} \dr{\al_j} \cL_j^{[\alpha]*}\lceil_{\al=0}(\un) = \delta_{jk}\beta_j^2T \cD_j(H_\cS',H_\cS'), 
\end{equation}
where we have used that $[\cL_j^*(H_\cS'),H_\cS']=0$ in the second identity, see Remark \ref{rem:effecthaminvariance}. 

Finally (\ref{eq:Lalpharacyderiv1})-(\ref{eq:Lalphacyderiv2}) follow directly from (\ref{eq:Lalphaderiv}) and (\ref{def:deformedrdmcyra}).
\end{proof}

\medskip

\noindent {\bf End of the proof of} (\ref{first moment})-(\ref{second momentrd}).

\noindent Since $\ds \dr{\al_j} r^\s(0)=\ds \dr{\al_j} E_1^\s(0)$ combining (\ref{eq:E1deriv}) and (\ref{eq:Lalpharacyderiv1}), recall also Remark \ref{rem:fluxmap}, we get (\ref{first moment}). 

We now turn to (\ref{second moment})-(\ref{second momentrd}). Using (\ref{eq:Lalpharacyderiv1}), the first term on the right-hand side of (\ref{eq:Enderivative2}) becomes
\[
I  = \beta_j\beta_k\tau_\s^2\, \oint_{\Gamma} \, (z-1)^{-1} \rho_{+}^\s \left(\varphi_k^\s\, \circ \left(z-\Lj_{\s}^*\right)^{-1} \left(\Phi_j^\s\right)\right)  \frac{{\rm d} z}{2\pi i},
\]
where $\tau_{cy}=T$ and $\tau_{ra}=\frac{T}{M}$. Recall $P_\s^{[0]}$ denotes the eigenprojection of $\cL_\s^*$ associated to the eigenvalue $1$. One has 
\[
\oint_{\Gamma} \, (z-1)^{-1} \rho_{+}^\s \left(\varphi_k^\s\, \circ \left(z-\Lj_{\s}^*\right)^{-1} \circ P_\s^{[0]} \left(\Phi_j^\s\right)\right)  \frac{{\rm d} z}{2\pi i} = \oint_{\Gamma} \, (z-1)^{-2} \rho_{+}^\s \left(\varphi_k^\s\, \circ P_\s^{[0]} \left(\Phi_j^\s\right) \right) \frac{{\rm d} z}{2\pi i} =0.
\]
Hence
\begin{eqnarray*}
I & = & \beta_j\beta_k\tau_\s^2\, \oint_{\Gamma} \, (z-1)^{-1} \rho_{+}^\s \left(\varphi_k^\s\, \circ \left(z-\Lj_{\s}^*\right)^{-1} \circ (Id-P_\s^{[0]}) \left(\Phi_j^\s\right)\right) \frac{{\rm d} z}{2\pi i} \\
 & = &  \beta_j\beta_k\tau_\s^2\, \rho_{+}^\s \left(\varphi_k^\s\, \circ \left(1-\Lj_{\s}^*\right)^{-1} \circ (Id-P_\s^{[0]}) \left(\Phi_j^\s\right)\right),
\end{eqnarray*}
where we have used that $\left(z-\Lj_{\s}^*\right)^{-1} \circ (Id-P_\s^{[0]})$ is regular at $z=1$. Moreover, since the spectral radius of $\cL_\s^*$ restricted to ${\rm Ran}(Id-P_\s^{[0]})$ is strictly less than one due to Assumption (\nameref{ergodicity}) we can write
\[
\left(1-\Lj_{\s}^*\right)^{-1} \circ (Id-P_\s^{[0]}) \left(\Phi_j^\s\right) = \sum_{n=0}^{\infty} \Lj_{\s}^{*n}\left( (Id-P_\s^{[0]})(\Phi_j^{\s})\right) =  \sum_{n=0}^{\infty} \Lj_{\s}^{*n}\left( \Phi_j^\s - \rho_+^\s(\Phi_j^{\s})\right), 
\]
so that
\[
I = \beta_j\beta_k\tau_\s^2\, \sum_{n=0}^{\infty}  \rho_{+}^\s \left(\varphi_k^\s\, \circ \Lj_{\s}^{*n}\left( \Phi_j^\s - \rho_+^\s(\Phi_j^{\s})\right) \right).
\]
Proceeding in the same way with the second term $II$ in (\ref{eq:Enderivative2}) we obtain that $I+II$ indeed corresponds to the infinite sums in (\ref{second moment})-(\ref{second momentrd}).

Finally,  the third term on the right-hand side of (\ref{eq:Enderivative2}) is $\ds III =   \rho_{+}^\s \left(\dr{\al_j}\dr{\al_k} \Lj_{\s}^{[\al]*}\lceil_{\al=0}(\un)\right)$, and using (\ref{eq:Lalpharaderiv2})-(\ref{eq:Lalphacyderiv2}) it is easy to see that it leads to the remaining terms in (\ref{second moment})-(\ref{second momentrd}).


\subsection{Proof of Theorem \ref{flucresults}}\label{prooffr}

The proof of 1)-3) is as direct application of the Gartner-Ellis and Bryc Theorems, see e.g. \cite{DZ,Ellbook,Bryc}, and the following lemma, while (\ref{flucrel}) and (\ref{energyconservation}) are direct consequences of (\ref{ES}) and (\ref{ts}) respectively.

\begin{lemm} \label{ldlemma} Suppose Assumption (\nameref{ergodicity}) holds. Then $e^{\s}(\al)$ is a well defined real analytic function on $\mb{R}^M$ and for any $\al\in\R^M$ and initial state $\rho$ one has
\begin{equation}\label{eq:lcgf_lim}
\lim_{n\to\infty} \frac{1}{n} e_{n,\rho}^{\s}(\al) = e^\s(\al).
\end{equation}
Moreover there exists a neighborhood $B$ of $0$ in $\C^M$ on which (\ref{eq:lcgf_lim}) holds. 
\end{lemm}

\begin{proof} The maps $\al\mapsto \cL_\s^{[\al]*}$ are clearly analytic on $\C^M$. Then Assumption (\nameref{ergodicity}) and Proposition \ref{deformedsp} guarantee that, for any $\al\in\R^M$, $r^{\s}(\al)$ is a positive isolated simple eigenvalue of $\cL_\s^{[\al]*}$. Hence regular perturbation theory ensures that $r^{\s}(\al)$ defines a real analytic map on $\R^M$ with positive values. This proves that $e^\s(\al)=\frac{1}{\tau_\s}\log r^{\s}(\al)$ is well defined and real analytic. Eq. (\ref{eq:lcgf_lim}) then follows from 1) in Theorem \ref{theojpw}.

The extension to a complex neighborhood of $0$ follows also by a standard perturbation theory argument, see also \cite{JPW}. Indeed, $\cL_\s^{[0]*}$ has a simple dominant eigenvalue $r^\s(0)=1$ so there exists $\delta,\varepsilon>0$ such that $\sp\left(\cL_\s^{[\al]*}\right)\setminus \{r^\s(\al)\} \subset \big\{ z\in \C\, \big| \, |z|<|r^\s(\al)|-\delta \big\}$ for any $\al\in\C^M$, $|\al|<\varepsilon$. Hence for such $\alpha$'s we have 
\[
r_{n,\rho}^\s(\al) = \rho\left( \left(\Lj_{\s}^{[\al]*}\right)^n (\un) \right) =  r^\s(\al)^n \left[ \rho\left(P_\s^{[\al]}(\un)\right) + O\left( \left(1-\frac{\delta}{|r^\s(\al)|}\right)^n \right) \right],
\]
where $P_\s^{[\al]}$ denotes the eigenprojection of $\cL_\s^{[\al]*}$ associated to $r^\s(\al)$. In particular $P_\s^{[\al]}(\un) = P_\s^{[0]}(\un)+O(\varepsilon) = \un+O(\varepsilon)$,  see Section \ref{proofmoments}, so that for all $n$
\[
r_{n,\rho}^\s(\al) =   r^\s(\al)^n \left[ 1+O(\varepsilon) + O\left( \left(1-\frac{\delta}{|r^\s(\al)|}\right)^n \right) \right]
\]
which proves the result.
\end{proof}


\end{document}